\documentclass{article}

\usepackage{arxiv} 
\usepackage[T1]{fontenc}    
\usepackage[utf8]{inputenc} 

\usepackage[english]{babel}
\usepackage[numbers]{natbib}  

\setcitestyle{authoryear,open={(},close={)}}

\usepackage{graphicx} 
\usepackage{enumerate}
\usepackage{multirow}
\usepackage{ulem}
\usepackage{hyperref}       
\usepackage{url}            
\usepackage{hyperref}
\hypersetup{
	colorlinks=true,
	linkcolor=blue,
	filecolor=blue,
	citecolor = blue,      
	urlcolor=cyan,
}
\usepackage{booktabs}       
\usepackage{multicol}
\usepackage{amsmath,amsfonts,amssymb}
\usepackage{comment}

\usepackage{amsthm}         
\numberwithin{equation}{section}

\usepackage{nicefrac}       
\usepackage{microtype}      
\usepackage{lipsum}
\usepackage{float} 
\usepackage{bm}
\usepackage{adjustbox}
\usepackage[linesnumbered,lined,boxed,commentsnumbered]{algorithm2e}
\usepackage{soul}
\usepackage{makecell}

\newcommand{\twee}{\operatorname{Tweedie}}

\newcommand{\diag}{\operatorname{diag}}

\newcommand{\Esp}[1]{\mathrm{E}\! \left[ #1 \right]}
\newcommand{\Var}[1]{\mathrm{Var}\! \left[ #1 \right]}

\newcommand{\indep}{\perp \!\!\! \perp}

\newtheorem{definition}{Definition}[subsection]
 
\newtheorem{prop}{Property}[subsection] 

\newtheorem{proposition}{Proposition}
\newtheorem{example}{Example}

\usepackage{caption}

\usepackage{tikz}
\usetikzlibrary{shapes,arrows}

\usetikzlibrary{decorations.pathreplacing}

\title{Comparison of Offset and Ratio Weighted Regressions in Tweedie Models with Application to Mid-Term Cancellations}

\author{
	Jean-Philippe Boucher \& Raïssa Coulibaly \\
	Chaire Co-operators en analyse des risques actuariels\\
	Département de mathematiques\\
	UQAM \\
}

\begin{document}
\newpage
\maketitle
	
\begin{abstract}
In property and casualty insurance, particularly in automobile insurance, risk exposure is commonly assumed to be proportional to the duration of coverage. This assumption leads to two standard estimation strategies: the ratio approach, which normalizes the response variable (e.g., claim cost or premium) by the exposure, and the offset approach, which incorporates a transformation of the exposure—typically its logarithm—as a fixed regressor in the mean structure of the model. Although both approaches rely on the same proportionality assumption, they are not equivalent when the response variable follows a Tweedie distribution, a framework widely used in insurance analytics. In this paper, we show that each approach can be implemented independently and yields a consistent estimator of the true mean parameter vector. We then show that the offset approach is asymptotically more efficient than the ratio approach, a result established both theoretically and through simulation studies.
However, when evaluated from the perspective of portfolio-level financial balance, the ratio approach exhibits superior performance, particularly in the presence of heterogeneous or truncated exposures arising from mid-term policy cancellations. These theoretical results are illustrated through an empirical analysis of an automobile insurance portfolio with a high cancellation rate, highlighting the practical implications of model choice for premium estimation under variable exposure conditions.
\end{abstract}

\keywords{Tweedie distribution, weighted regression, offset variable, policy cancellation, positive definite matrix}

\newpage 
\section{Introduction}\label{introduction}

In insurance, particularly within the property and casualty sector like automobile insurance, risk exposure has traditionally been associated with the duration of coverage (see \cite{denuit2007actuarial}). Risk exposure refers to the extent to which an insured entity is subject to potential claims or losses over the policy period. For example, under a linear interpretation of exposure, a one-year policy would be considered twice as risky as a six-month policy, assuming identical characteristics for the insured. This concept is central to many actuarial research projects and practical applications.

Given its importance, it becomes essential to develop a comprehensive framework for incorporating various dimensions of risk exposure into the modelling of random variables used for ratemaking. Such models ensure that the premiums charged accurately reflect the underlying risks.

More formally, we are interested in modelling the random variable \(Y_i\), representing either the total number of claims or the total claims amount (also referred to as the loss cost) for contracts \(i = 1, \dots, n\), where \(n\) denotes the total number of contracts in the portfolio. Assuming that the expected value corresponds to the premium (as discussed in \cite{denuit2007actuarial} and \cite{frees2014predictive}), the premium can be expressed as:

\begin{equation}\label{annualprem}
    \Esp{Y_i|\mathbf{X_i}} = g^{-1}\left(\mathbf{X}_i^\top \bm{\beta} \right),
\end{equation}

where \(\bm{\beta}\) is a parameter vector to be estimated, and \(\mathbf{X_i} = (x_{i,0}, \dots, x_{i,q})^\top\) represents the vector of risk characteristics associated with  contract \(i\) (with \(q\) being a positive integer). These risk characteristics—such as age, sex, and the value of the insured goods—are used as covariates in the regression model, allowing the premium to vary based on these factors. The function \(g(\cdot)\) serves as a link function that connects the score \(\mathbf{X}_i^\top \bm{\beta}\) to the expected value, similar to the approach in Generalized Linear Models (GLMs) (see \cite{nelder1972generalized}).

An analysis of historical motor insurance data reveals significant variability in exposure periods among insured individuals. This variation often occurs when coverage is cancelled before the contract ends due to reasons such as the sale or replacement of a vehicle, theft, or a total loss resulting from an accident. For modelling purposes, insurers typically treat the exposure period as a fraction of the total number of days in a year. This fraction, referred to as \textbf{risk exposure} and denoted by \(t\), is a positive value between 0 and 1, reflecting the proportion of the year that the contract was active.

Formally, we consider a portfolio of \(n\) insured contracts, indexed by \(i\), some of which have a risk exposure, denoted by \(t_i\),  that is less than 1. To account for this variation in exposure when calculating premiums and estimating model parameters, it is necessary to generalize Equation (\ref{annualprem}). This paper examines two approaches within the maximum likelihood estimation framework to address this generalization:

\begin{enumerate}
\item The first approach assumes that the exposure \(t_i\) is proportional to the mean of the conditional distribution \(Y_i|\mathbf{X_i}\), where \(f_{Y_i}(\cdot)\) represents the density or probability mass function of \(Y_i|\mathbf{X_i}\), and \(y_i\) denotes the observed loss cost for contract \(i\). Under this assumption, the log-likelihood function for estimating the parameter vector \(\bm{\beta}\) is expressed as:

\[
\ell(\boldsymbol{\beta}|y_1,\ldots,y_n) = \sum_{i=1}^{n} \log\left( f_{Y_i}(y_i) \right),
\]

where the mean of \(Y_i\) being linearly proportional to \(t_i\), such that:

\[
\Esp{Y_i|\mathbf{X_i}} = g^{-1}\left(\mathbf{X}_i^\top \bm{\beta} \right) \times t_i.
\]

In this framework, \(g(\cdot)\) is typically chosen as a logarithmic link function, resulting in:

\[
\exp\left(\mathbf{X}_i^\top \bm{\beta} \right) t_i = \exp\left(\mathbf{X}_i^\top \bm{\beta} + \log(t_i)\right),
\]

where \(\log(t_i)\) acts as an offset variable. This method is thus referred to as the \textbf{offset approach}.

\item The second approach, referred to as \textbf{ratio approach}, focuses on annualizing the random variable \(Y_i\) to facilitate direct comparisons of policyholders' claims on an annual basis. In this method, a new random variable \(Z_i = \frac{Y_i}{t_i}\) is defined for each contract \(i = 1, \dots, n\), resulting a normalized set of claim amounts \(z_1, \dots, z_n\). 

To distinguish it from the offset approach, the mean of \(Z_i|\mathbf{X_i}\) is assumed to no longer depend on \(t_i\), and the density or probability mass function of \(Z_i|\mathbf{X_i}\) incorporates the risk exposure through this weighting. The log-likelihood function for estimating \(\bm{\beta}\) is expressed as:

\[
\ell(\boldsymbol{\beta}|z_1,\ldots,z_n) = \sum_{i=1}^{n}t_i \log\left(f_{Z_i}(z_i)\right),
\]

where \(f_{Z_i}(\cdot)\) is the density or probability mass function of \(Z_i|\mathbf{X_i}\), and the mean is given by:

\[
\Esp{Z_i|\mathbf{X_i}} = g^{-1}\left(\mathbf{X}_i^\top \bm{\beta} \right),
\]

where \(g(\cdot)\) a logarithmic link, the expected loss cost for an insured \(i\) exposed to risk over a period of \(t_i\) years is then expressed:

\[
\Esp{Y_i|\mathbf{X_i}} = \Esp{t_i Z_i|\mathbf{X_i}} = \exp\left(\mathbf{X}_i^\top \bm{\beta} \right) \times t_i = \exp\left(\mathbf{X}_i^\top \bm{\beta} + \log(t_i)\right),
\]

which has the same form as the expected loss cost from the offset approach.
     
\end{enumerate}

It is important to emphasize that in the offset approach, the random variable \(Z_i\) can be used to estimate the parameter vector \(\bm{\beta}\), but this method should not be confused with the ratio approach. Similarly, in the ratio approach, the random variable \(Y_i\) can also be used, but this method should not ne mistaken for the offset approach. 

We refer to the offset approach when risk exposure is only incorporated into premium modelling as an offset variable. Conversely, we refer to the ratio approach when risk exposure is only included as a weight in the premium modelling process. These two approaches represent distinct methods for integrating risk exposure, ensuring that premiums are appropriately aligned with the insured's actual level of risk.

\subsection{Examples for claim counts} \label{introduction.1}

To clarify the issue, we can illustrate the situation with a preliminary comparison of the two approaches,  assuming that \( Y_i \) represents the number of claims for insured \( i \).

\begin{example} \label{exPoisson}
Under the ratio approach, we assume that \( Z_i,  i=1, \ldots, n \) follows a Poisson distribution with mean \(\zeta_i = \exp(\mathbf{X}^\top_i \boldsymbol{\beta})\). In this framework, it can be shown that the ratio approach corresponds to a Poisson weighted regression and yields the same inference results of the parameter vector \(\bm{\beta}\) as in the offset approach, where  \( Y_i,  i=1, \ldots, n \) is Poisson distributed with mean \( \mu_i = t_i \exp(\mathbf{X}^\top_i \boldsymbol{\beta}) = \exp(\mathbf{X}^\top_i \boldsymbol{\beta} + \log(t_i)) \).
\end{example}

\begin{proof}
The log-likelihood function of the offset and ratio approaches are given as follows:

\begin{itemize}
	\item \textbf{Offset Approach:}
	\begin{align*}   
		\ell(\boldsymbol{\beta}|y_1,\ldots,y_n) 
		&= \sum_{i=1}^{n} \log \left(\frac{\exp(-\mu_i) \mu_i^{y_i}}{y_i!} \right)\\
		&= \sum_{i=1}^{n} \left(-t_i \exp\left(\mathbf{X}^\top_i \boldsymbol{\beta}\right) + y_i \left(\mathbf{X}^\top_i \boldsymbol{\beta}\right) + y_i \log(t_i) - \log(y_i !)\right)  \\
		&\propto \sum_{i=1}^{n} \left(-\exp(\mathbf{X}^\top_i \boldsymbol{\beta}) + z_i \left(\mathbf{X}^\top_i \boldsymbol{\beta}\right)\right) \times t_i.
	\end{align*}
	
	\item \textbf{Ratio Approach:}
	\begin{align*}   
		\ell(\boldsymbol{\beta}|z_1,\ldots,z_n) 
		&= \sum_{i=1}^{n} t_i \log \left(\frac{\exp(-\zeta_i) \zeta_i^{z_i}}{z_i!} \right)\\
		&= \sum_{i=1}^{n} \left(-t_i \exp\left(\mathbf{X}^\top_i \boldsymbol{\beta}\right) + t_iz_i \left(\mathbf{X}^\top_i \boldsymbol{\beta}\right)  - \log((t_iz_i) !)\right)  \\
		&\propto \sum_{i=1}^{n} \left(-\exp(\mathbf{X}^\top_i \boldsymbol{\beta}) + z_i \left(\mathbf{X}^\top_i \boldsymbol{\beta}\right)\right) \times t_i.
	\end{align*}
\end{itemize}

For the Poisson distribution, the offset and the ratio approaches are equivalent, meaning there is no need to choose between them.
\end{proof}

\begin{example}
In the offset approach, we assume that \( Y_i,  i=1, \ldots, n \) follows a zero-inflated Poisson distribution characterized by the parameters \( \phi \) (for the zero-inflation) and \( \mu_i =  \exp(\mathbf{X}^\top_i \boldsymbol{\beta} + \log(t_i))\) (from the Poisson distribution). In the ratio approach, we assume that \( Z_i,  i=1, \ldots, n \) follows a zero-inflated Poisson distribution characterized by the parameters \( \phi \) (for the zero-inflation) and \( \zeta_i =  \exp(\mathbf{X}^\top_i \boldsymbol{\beta})\) (from the Poisson distribution). It can be shown that these two approaches yield different inference results for the estimation of the  parameter vector \(\bm{\beta}\).
\end{example}

\begin{proof}
For a portfolio of \( n \) contracts, we can develop the log-likelihood of the offset approach to see that both approaches are not equivalent:

\begin{align*}   
\ell(\boldsymbol{\beta}|y_1,\ldots,y_n) 
&= \sum_{i=1}^{n} \log \left(I(y_i=0) \phi + (1-\phi) \frac{\exp(-\mu_i) \mu_i^{y_i}}{y_i!} \right)\\
& \ne \sum_{i=1}^{n} \log \left(I(z_i=0) \phi + (1-\phi) \frac{\exp(-\zeta_i) \zeta_i^{z_i}}{z_i!} \right) \times t_i 
= \ell(\boldsymbol{\beta}|z_1,\ldots,z_n), 
\end{align*}

In contrast to the Poisson distribution, a decision must be made between the offset and ratio methods when we want to suppose that the number of claims follows this form of zero-inflated distribution.
\end{proof}

Depending on the distribution used for \(Y_i\), it is not straightforward to determine whether both approaches are equivalent. Even when the underlying distribution is assumed to be the same, the estimated parameters obtained from each approach may differ, potentially leading to different premiums for individual policyholders. Ideally, if the data-generating process for insurance claims were fully understood, one could determine which modelling approach is most appropriate. However, in practice, this is rarely the case, making it difficult to justify one formulation over the other on theoretical grounds alone. Consequently, the choice between approaches often relies on a combination of empirical performance and desirable mathematical properties.

In this context, it is worth noting that \cite{denuit2019effective} already provide some comparisons between these two approaches. In particular, they show that the offset approach can be interpreted as a {weighted regression} when the response variable follows a Tweedie distribution. Building on their work, we extend the comparison further by emphasizing the role of stochastic orders and the principle of financial balance, which provide additional insights into the structural differences between the two modeling approaches.

\subsection{Structure of the paper} \label{introduction.2}

This paper emphasizes the differences between the offset and ratio regression approaches in modeling contract loss costs using the {Tweedie distribution}, rather than focusing on claim frequency. Section~\ref{rap.twee} introduces the notation and model assumptions, detailing key properties of the Tweedie distribution and the parameter inference methods used. It also provides the technical framework required to establish the equivalence results and the asymptotic comparisons developed in the remainder of the paper.  Section~\ref{ofwe.twee} then explains the offset and ratio approaches, as well as the link between them. We show that both approaches can be viewed as weighted regressions, where only the form of the weight differentiates the approaches. Section~\ref{balance.Twee} presents an analysis of the consistency of the parameter vector estimators for both approaches, demonstrates the asymptotic efficiency of the offset-approach parameter vector estimator, and examines the gap between observed aggregate loss costs and the sum of estimated premiums based on simulated data. In Section~\ref{ap.num}, we provide an empirical illustration of both approaches using a sample of automobile insurance data from a Canadian insurer. Finally, Section~\ref{concl} concludes the paper.

\section{GLM with Tweedie-distributed Response Variable} \label{rap.twee}

The loss cost of a contract is the total amount of claims, calculated as the aggregate sum of individual claim costs. Formally, let $N$ denote the total number of claims for a contract, and let $Z_{k}$ represent the cost of the 
$k$-th claim, where $k, k=1,\ldots,N$ if $N>0$. Thus, assuming $Y$ denotes the loss cost of a contract, as $Y$ is obtained as follows:

\[Y=
\begin{cases} 
	\sum_{k=1}^{N}Z_{k}&\text{if  $N >0$}\\ \\
	0 & \text{if $N=0$.}
\end{cases}
\] 

Even though individual claim costs $Z_{1},\ldots,Z_{N}$ are continuous random variables, we can demonstrate that the loss cost $Y$ is not a continuous random variable. Specifically, the probability of a null loss cost is non-zero, indicating that the loss cost has a mass at zero and is therefore a semi-continuous random variable. As a result, insurers often use the Tweedie distribution to model a contract premium when the loss cost is the response variable. The rationale for using the Tweedie distribution is its ability to represent a mixture distribution (both continuous and discrete), as noted by \cite{tweedie1984index}. Additionally, \cite{delong2021making} provides a comprehensive summary of this mixture distribution, with relevant applications in actuarial science.

Accordingly, we model the premium using the loss cost as the response variable, assuming that the latter follows a Tweedie distribution. This distribution is viewed as a compound Poisson–Gamma mixture and belongs to the linear exponential family, which justifies the use of the Generalized Linear Model (GLM) framework (see Tweedie, 1984; Jørgensen, 1987). Before introducing the GLM approach with the Tweedie distribution, we briefly review its properties by referring to \cite{tweedie1984index}, \cite{jorgensen1987exponential}, and \cite{delong2021making}. It is also important to note that several results presented here can also be found in \cite{wuthrich2023statistical} and \cite{denuit2019effective}; however, this review is essential for a proper understanding of the offset and ratio regression approaches.

\subsection{Exponential Family Form of Tweedie Distribution} \label{intro.1}

\begin{definition} \label{TweedieDef}
The density function of the Tweedie distribution as a mixture distribution is given by: 

	\begin{equation} \label{eq.twee.2}
		f(y|\mu,w,\phi)= \exp\left(\frac{w}{\phi} \left(\frac{\mu^{1-p}}{1-p}y   -  \frac{\mu^{2-p}}{2-p} \right) + a(y,w,\phi) \right), y \in \mathbf{R},
	\end{equation}

where:

\begin{itemize}
        \item $\mu >0$ denotes the mean parameter;
	\item $\phi >0$ denotes the dispersion parameter;
	\item $w>0$ denotes the weight parameter;
	\item $p \in ]1,2[$ denotes the variance parameter; 
        \item $a(\cdot)$ denotes a function that does not depend on $\mu$.
\end{itemize}
\end{definition}

It is important to note that assuming the variance parameter lies within the interval \( ]1, 2[ \) ensures that the Tweedie distribution density is defined as described above. As noted by \cite{tweedie1984index} and \cite{jorgensen1987exponential}, the Tweedie distribution can be generalized for any positive variance parameter \( p > 0 \). However, they also mention that when the variance parameter lies within the interval \( ]0, 1[ \), the Tweedie distribution does not belong to the exponential family. The density function \( f(\cdot) \) mentioned above is expressed in the form of the exponential family density function. In the case of the Tweedie distribution, this family is sometimes referred to as the Exponential Dispersion Family (EDF) due to the inclusion of the dispersion parameter (see \cite{jorgensen1997theory}).

It should also be noted that \( f(\cdot) \) lacks a specific form because the function \( a(\cdot) \) is unknown. To determine \( a(\cdot) \), one can refer to \cite{delong2021making}, which provides a comprehensive overview of the Tweedie distribution. In this paper, we do not focus on determining \( a(\cdot) \) since it does not affect the regression approaches discussed here. Instead, we concentrate on the canonical parameter and the canonical link, which are relevant to insurance pricing. These two quantities for the Tweedie distribution are defined as follows:

\begin{itemize}
	\item \( \theta \) denotes the canonical parameter and is identified by: \( \theta = \frac{\mu^{1 - p}}{1 - p} \);
	\item  The canonical link function is defined by: \( s >0 \mapsto  \frac{s^{1 - p}}{1 - p}\).
\end{itemize}

It is important to note that the canonical link connects the canonical parameter to the mean parameter of distributions within the exponential family. In car insurance pricing, the canonical link function is often employed within the framework of GLM to ensure the balance property, as described by \cite{nelder1972generalized}. The balance property ensures that the sum of the estimated premiums equals the sum of the observed loss costs. The rationale behind this equality is that pricing models must replicate the observed risk at each portfolio level as accurately as possible. Consequently, the primary objective for insurers is to develop models that closely reproduce the observed risk across all portfolio levels.

However, when the response variable follows a Tweedie distribution, interpreting the premium through the canonical link becomes challenging. Therefore, as discussed in the following section, choosing the log link to model the premium instead of the canonical link offers a more straightforward interpretation of the premium. It is important to note, however, that using the log link in premium modelling introduces a gap between the sum of the observed loss costs and the sum of the estimated premiums. This paper will analyze this gap in detail.

For the remainder of this paper, we denote the parameters of the Tweedie distribution by \( (\mu,w,\phi, p) \) and we define \( \mathcal{T} \) as the set of Tweedie distributions \( (\mu, w, \phi, p) \) where the variance parameter \( p \)  lies within the interval \( ]1,2[ \).

\subsection{Notations}\label{rap.twee.2.0}	
We consider \( n \) independent insurance contracts. For each contract \( i \) (where \( i = 1,\ldots,n\)),  the following variables are defined:
\begin{itemize}
    \item \( Y_i \): A random variable representing the loss cost for contract \( i \).
    \item  \( y_i \): The observed loss cost for contract \( i \).
    \item \( t_i \) The risk exposure for contract \( i \).
\end{itemize}

Additionally, we consider \( q \) risk factors, denoted by \( \bm{x_1},\ldots,\bm{x_q} \). For each contract \( i \), the following are also defined:

\begin{itemize}
    \item \( x_{i,j} \): The observed characteristic for contract \( i \) corresponding to risk factor \( \bm{x_j} \), where \( j = 1,\ldots,q\).
    \item  \( \mathbf{X}_i\): The vector of risk characteristics associated with contract \( i \), defined as \( \mathbf{X}_i = \left(x_{i,0}, x_{i,1},\ldots, x_{i,q}\right)^\top\), where \(x_{i,0} = 1, \forall i=1,\ldots,n\).
\end{itemize}

Finally, the matrix \( \bm{X} \), representing the risk factors, is given by:

$$ \bm{X}=\begin{pmatrix} 1& x_{1,1}&\ldots& x_{1,q}\\
                   1& x_{2,1}&\ldots& x_{2,q}\\
                   \vdots& \vdots& \ddots&\vdots\\
                   1& x_{n,1}&\ldots& x_{n,q}
  \end{pmatrix}$$

This matrix \( \bm{X} \) captures all the observed risk characteristics across the contracts.

\subsection{Assumptions and properties}\label{rap.twee.2.1}
The main assumptions of GLM are as follows:

\begin{enumerate}
    \item The distribution of \( Y_{i}| \mathbf{X}_{i}\) belongs to the exponential linear family for all contract \( i \) where \( i=1,\ldots,n \).
    \item  \(  Y_{i_1}| \mathbf{X}_{i_1}\indep  Y_{i_2}|\mathbf{X}_{i_2}, \forall i_1,i_2 = 1,\ldots,n ~ \text{such that}~i_1 \ne i_2\).
     \item \( \bm{X} \) has full rank $q + 1$.
\end{enumerate}

The first assumption in the case of the Tweedie distribution is equivalent to stating that \( Y_{i}| \mathbf{X}_{i}\) follows a Tweedie distribution with parameters \((\mu_i,w_i,\phi, p)\). The second assumption implies the independence of the loss costs between contracts, given the risk characteristics. It is important to note that these first two assumptions are fundamental for constructing the likelihood function. The third assumption is essential for calculating the variance of the estimator of \( \mu_i \). 

The concept of GLM revolves around statistical inference on \( \mu_i \), which is considered the premium for contract \( i \) in insurance pricing. This consideration is justified by the fact that the mean parameter \( \mu_i \), regardless of the distribution of \( Y_i | \mathbf{X}_i \), corresponds to the constant closest to \( Y_i | \mathbf{X}_i \) in terms of Euclidean distance (for further details, see \cite{denuit2007actuarial}). 

Thus, for modelling the premium of contract \( i \), we assume that the logarithm of this premium (\( \mu_i \)) linearly depends on the characteristics vector \( \mathbf{X}_i\) as follows:

\begin{align} \label{rap.twee.eq.1}
    \mu_i &=  \exp\left(\bm{X}^\top_i\bm{\beta}\right)= \exp\left(\beta_0\right) \prod_{j=1}^q \exp\left(\beta_j x_{i,j}\right) = b_0 \times \prod_{j=1}^q r_j,
\end{align}

where:

\begin{itemize}
    \item \( \bm{\beta} \): The parameter vector, defined as \( \bm{\beta} = (\beta_0,\beta_1,\ldots,\beta_q)^\top\) where \(\beta_j \in \mathbf{R}\) for \( j = 0,1,\ldots,q\);
    \item \( b_0 \): The basic premium, defined as  \( b_0 = \exp\left(\beta_0\right)\). 
    \item \( r_j \): The relativity factors associated with each characteristic, defined as \( r_j = \exp\left(\beta_j x_{i,j}\right)\).
\end{itemize}

The advantage of using the logarithm function as the link function is that the premium \( \mu_i \) for contract \( i \) can be expressed as the product of the basic premium \( b_0 \)  and relativity factors \( r_j, j=1, \ldots, q\). While the basic premium \( b_0 \) is consistent across all contracts, the relativity factors depend on the specific risk characteristics of each contract. Consequently, statistical inference on \( \mu_i \) is equivalent to estimating the parameter vector \( \bm{\beta} \). 

To conclude, it is pertinent to introduce an interesting property of the Tweedie distribution, as discussed in \cite{denuit2019effective}, which will be valuable for further developments.

\begin{prop}[Tweedie Invariance Property] \label{rap.twee.prop}
If a random variable \( Z \) is Tweedie-distributed with parameters \( (\mu, w, \phi, p) \) and belongs to the set $\mathcal{T}$, then \( tZ \), for any positive \( t \), is also Tweedie-distributed with parameters \( (t\mu, \frac{w}{t^{2-p}}, \phi, p) \) and belongs to the same set $\mathcal{T}$. Conversely, if \( tZ \), for any positive $t$, is Tweedie distributed with parameters \( (\mu, w, \phi, p) \) and belongs to the set $\mathcal{T}$, then \( Z \) is also Tweedie-distributed with parameters\( (\frac{\mu}{t}, wt^{2-p}, \phi, p) \) and belongs to the same set $\mathcal{T}$. Formally, we have:

$$Z \in \mathcal{T} \Longleftrightarrow t Z \in \mathcal{T}, \forall t >0.$$
\end{prop}

\subsection{Log-likelihood Function}\label{rap.twee.3}	
The GLM employs the maximum likelihood approach to estimate the parameter vector $\bm{\beta}$ \citep{nelder1972generalized}. The maximum likelihood approach maximizes the likelihood function, which depends on the parameter vector $\bm{\beta}$ given the observations $y_1,\ldots,y_n$. For a detailed discussion on the theory of maximum likelihood, see  \cite{casella2024statistical}. In the case of the Tweedie distribution, by utilizing the conditional independence between contracts and Definition \ref{TweedieDef}, the likelihood function is expressed as follows:

$$L(\beta_0,\ldots,\beta_q) = \prod_{i=1}^n \exp\left(\frac{w_i}{\phi} \left(\frac{\mu^{1-p}}{1-p}_iy_i   -  \frac{\mu^{2-p}}{2-p}_i \right) + a(y_i,w_i,\phi) \right).$$
    
In practice, the maximum likelihood approach often uses the logarithm of the likelihood function as the objective function instead of the likelihood function itself. This is because a function and all strictly increasing transformations of that function reach their maximum at the same point (See \cite{boyd2004convex} for the theory on function optimization). Therefore, in the Tweedie case, the logarithm of the likelihood function (commonly referred to as the log-likelihood function) is given by (\ref{loglik.twe.1}). 

\begin{align} \label{loglik.twe.1}
    \ell(\beta_0,\ldots,\beta_q)&= \sum_{i=1}^n \left(\frac{w_i}{\phi} \left(\frac{\mu^{1-p}_i}{1-p}y_i   -  \frac{\mu^{2-p}_i}{2-p} \right) + a\left(y_i,w_i,\phi,p\right)\right)\\\label{loglik.twe.2}
    &\propto \frac{1}{\phi} \sum_{i=1}^n w_i \left(\frac{\mu^{1-p}_i}{1-p}y_i   -  \frac{\mu^{2-p}_i}{2-p} \right) .
\end{align}

Additionally, considering equation (\ref{loglik.twe.2}), we note that maximizing the log-likelihood in the Tweedie case is equivalent to performing a weighted regression, where each contract \( i \) is weighted by \( w_i \). The log-likelihood maximization can also be interpreted as finding the zeros of the gradient function associated with the model. The gradient function is the vector of the first partial derivatives of the log-likelihood function (\( \frac{\partial \ell(\beta_0,\ldots,\beta_q)}{\partial \beta_j}, j=0,\ldots,q\)). For the Tweedie case, the gradient function is derived as follows:

\begin{align*}
	\nabla(\bm{\beta})&=  \left( \frac{\partial \ell(\beta_0,\ldots,\beta_q)}{\partial\mu_i}\frac{\partial \mu_i}{\partial \beta_0},\ldots,   \frac{\partial \ell(\beta_0,\ldots,\beta_q)}{\partial \mu_i}\frac{\partial \mu_i}{\partial \beta_j}\right)^{T}\\
	&=\frac{1}{\phi}\left( \sum_{i =1}^n w_i\frac{y_i - \mu_i}{\mu^{p -1}_i},\ldots, \sum_{i =1}^n w_i\frac{y_i - \mu_i}{\mu^{p -1}_i}x_{i,q} \right)^\top\\
	&=\frac{1}{\phi} \bm{X}^\top\bm{D}\bm{R};
\end{align*}

The vector \( \bm{R} \) and the matrix \( \bm{D} \) are defined as follows:

\begin{align}
\bm{R}= \left(\frac{y_i}{\mu_i} - 1\right)_{i= 1,...,n} \  
\bm{D}= \diag\left(w_i \mu^{2 - p}_i \right)_{i= 1,...,n} \label{MatrixD}.   
\end{align}

\subsubsection{Iterative Algorithm}\label{rap.twee.4}

Generally, GLM theory adopts the IRLS (Iteratively Reweighted Least Squares) algorithm for maximizing the log-likelihood, which is a variant of the Newton-Raphson algorithm (For a detailed discussion on the Newton-Raphson algorithm, refer to \cite{boyd2004convex}). It is important to note that the IRLS algorithm is based on the Fisher Information Matrix, the theory of which can be found in \cite{casella2024statistical}. In this section, we review the Fisher Information Matrix in the context of the Tweedie distribution and justify its invertibility. We also recall the definition of positive definite matrices and their properties, as discussed by \cite{horn2012matrix}.

Let \( \bm{\beta}_{(k)} \) denote the estimate of  \( \bm{\beta} \) at iteration \( k \), The IRLS algorithm is defined as follows:

\begin{equation} \label{irls.algorithm.1}
    \bm{\beta}_{(k + 1 )} = \bm{\beta}_{(k)} + \left(I_n\left(\bm{\beta}_{(k)}\right)\right)^{-1} \nabla\left({\bm{\beta}_{(k)}}\right);
\end{equation}

\( I_n(\cdot) \) is the Fisher information matrix and is defined as follows:

\begin{align*}
	I_n(\bm{\beta}) &= - \Esp{\text{Hess}\left({\bm{\beta}}\right)} \\
	&=  \frac{1}{\phi}\left(\sum_{i=1}^n w_i \mu^{2 - p}_i x_{i,j_1}x_{i,j_2} \right)_{_{j_1,j_2= 0,\ldots,q}} \\
	&= \frac{1}{\phi} \bm{X}^\top\bm{D}\bm{X},
\end{align*}

where \( \text{Hess}\left({\bm{\beta}}\right) \) is the Hessian matrix. Referring to equation (\ref{irls.algorithm.1}),  we note that the IRLS algorithm relies on the invertibility of the Fisher Information Matrix. The inverse of the Fisher Information Matrix is also used to calculate the asymptotic covariance matrix of the parameter vector estimator. As demonstrated in this paper, comparing the covariance matrices of the parameter vector estimators is crucial for comparing premium estimators in the offset and ratio approaches. Therefore, analyzing the invertibility of the Fisher Information Matrix is of prime importance.

We now review the properties of positive definite matrices and introduce the following notations:

\begin{itemize}
  \item $\mathcal{M}_{n,m}$ denotes the set of matrices with $n$ rows and $m$ columns.
  \item $\mathcal{P}_m$ denotes the set of positive definite matrices.
\end{itemize}

\begin{definition} \label{PosMt.def}
A matrix $M \in \mathcal{M}_{m,m}$ is positive definite ($M\succ 0$), if the following condition holds: 

$$ x^\top M x > 0; \forall x \in \mathbb{R}^m.$$
\end{definition}

\begin{prop} \label{PosMt.pro}
With $\mathcal{M}_{n,m}$ and $\mathcal{P}_m$ defined above, the following relationships can be stated:

\begin{align}\label{PosMt.pro.1}
	A \in \mathcal{P}_m &\Longleftrightarrow A^{-1} \in\mathcal{P}_m;\\\label{PosMt.pro.2}
	\forall A,B \in  \mathcal{P}_m: A - B\succ 0 &\Longleftrightarrow  B^{-1} - A^{-1} \succ 0;\\\label{PosMt.pro.3}
	\forall A \in \mathcal{P}_m, \forall X \in \mathcal{M}_{m,q} : X^\top AX \succ 0 &\Longleftrightarrow X~ \text{has rank} ~q.
	\end{align}
\end{prop}		

Property (\ref{PosMt.pro.1}) serves as a necessary and sufficient condition for proving the invertibility of a matrix. Thus, to prove the existence of the inverse of the Fisher Information Matrix, we must show that this matrix is positive definite.  Using the GLM theory assumption regarding the rank of the matrix \( \bm{X} \) and the necessary and sufficient condition given in (\ref{PosMt.pro.3}), we conclude that the Fisher Information Matrix is indeed positive definite. The necessary and sufficient condition given in (\ref{PosMt.pro.2}) is also helpful when comparing premium estimators in the offset and ratio approaches. 

\subsubsection{Remarks}
The variance and weight parameters exclusively influence the estimates of the parameter vector in the IRLS algorithm. However, as stated in equation (\ref{irls.algorithm.2}), the dispersion parameter does not affect these estimates, and therefore, it is not a focus of this paper. Finally, if the weight of each contract is known and the variance parameter is also known, the GLM in the Tweedie case involves using the IRLS algorithm to obtain an estimate of the parameters vector.

\begin{equation} \label{irls.algorithm.2}
  \left( I_n(\bm{\beta})\right)^{-1} \nabla\left(\bm{\beta}\right) =  \left( X^\top DX\right)^{-1} X^\top DR.
\end{equation}

It is also noteworthy that the IRLS algorithm is implemented in the R software through the \textit{glm} function.

\section{The Offset and the Ratio Approaches} \label{ofwe.twee}

To present the offset and ratio regression approaches within the context of a Tweedie distribution and to examine their similarities and differences, we continue to consider \( n \) independent contracts, along with the same notations and assumptions as in the previous section. The primary distinction from the previous section is that, in these two regression approaches, the premium of contract \( i \) is adjusted to account for its risk exposure \( t_i \).

\subsection{Offset Approach} \label{ofwe.twee.1}
As mentioned in Section \ref{introduction.1}, the statistical inference on the parameter vector  \( \bm{\beta} \) in offset approach is performed with the loss costs \( Y_1,\ldots,Y_n \). To account for the logarithm of a contract's risk exposure as the offset variable, we assume that \( Y_i | \mathbf{X}_i \) follows the Tweedie distribution with parameters \( (\mu_i, 1, \phi, p) \) where \( \mu_i \) is defined as $\mu_i = t_i \exp\left(\mathbf{X}^\top_i \bm{\beta} \right)$. We also denote an estimate of \( \mu_i \) in the offset regression by \( \widehat{\mu}^O_i \), obtained as follows:

\begin{equation} \label{ofwe.twee.eq.2}
	\widehat{\mu}^O_i = \exp\left(\mathbf{X}^\top_i \widehat{\bm{\beta}}^O+ \log(t_i) \right),
\end{equation}

where \( \widehat{\bm{\beta}}^O \), an estimate of \( \bm{\beta} \) in offset approach, is obtained using the IRLS algorithm with the following quantities:

\begin{itemize}
    \item \textbf{Log-likelihood Function:}
     \begin{equation} \label{ofwe.twee.eq.3}
         \ell^{O}(\beta_0,\ldots,\beta_q)= \sum_{i=1}^n \left(\frac{1}{\phi} \left(\frac{\mu^{1-p}_i}{1-p}y_i   -  \frac{\mu^{2-p}_i}{2-p} \right) + a^{O}\left(y_i,\phi,p\right)\right).
     \end{equation}
     \item \textbf{Gradient Function:}
     \begin{align} \label{ofwe.twee.eq.4}
\nabla^{O}(\bm{\beta}) &= \frac{1}{\phi}\left( \sum_{i =1}^n \frac{y_i - \mu_i}{\mu^{p - 1}_i},\ldots, \sum_{i =1}^n \frac{y_i - \mu_i}{\mu^{p - 1}_i}x_{i,q} \right)^\top\\\label{ofwe.twee.eq.5}
	&= \frac{1}{\phi} \bm{X}^\top\bm{D}^{O}\bm{R}^{O};
\end{align} 
The vector \( \bm{R}^{O} \) and the matrix \( \bm{D}^{O} \) are defined as follows:

$$\bm{D}^{O}= \diag\left(t^{2 - p}_i  \exp\left((2 - p)\mathbf{X}^\top_i \bm{\beta} \right)   \right)_{i= 1,...,n}, \bm{R}^{O}= \left(\frac{y_i}{\mu_i} - 1\right)_{i= 1,...,n}.$$
\end{itemize}

\subsection{Ratio Approach} \label{ofwe.twee.2}

We saw in Section \ref{introduction.1} that the ratio-based method can be interpreted as a form of weighted regression. In this framework, the response variable is defined as the ratio of the contract's loss cost to its risk exposure, with each contract being assigned a weight proportional to its respective risk exposure. We denote the ratio of contract \( i \)'s  loss cost \( Y_i\)  and its risk exposure \( t_i \) by \( Z_i \), defined as \( Z_i = \frac{Y_i}{t_i}\), for \(i=1,\ldots,n\). To account for the risk exposure \( t_i \) as the weight of contract \( i \), we also assume \( Z_i| \mathbf{X}_i \) follows a Tweedie distribution with parameters \( (\zeta_i, t_i, \phi, p) \), where the mean parameter \( \zeta_i \) is assumed to be:

\begin{equation} \label{ofwe.twee.eq.6}
\zeta_i = \exp\left(\mathbf{X}^\top_i \bm{\beta} \right), 
\end{equation}

and \( \widehat{\zeta_i}  \) is an estimate of \( \zeta_i \), equal to \( \widehat{\zeta_i}  = \exp\left(\mathbf{X}^\top_i \widehat{\bm{\beta}}^R \right)\). The vector \( \widehat{\bm{\beta}}^R \), an estimate of \( \bm{\beta} \) in ratio approach, is obtained using the IRLS algorithm with the following quantities:

\begin{itemize}
    \item \textbf{Log-likelihood Function:} This function is constructed with the observations \( z_1,\ldots,z_n  \) where \( \zeta_i = \frac{y_i}{t_i}, i=1,\ldots,n \), as follows:
    
     \begin{equation} \label{ofwe.twee.eq.8}
         \ell^R(\beta_0,\ldots,\beta_q)= \sum_{i=1}^n \left\{\frac{t_i}{\phi} \left(\frac{\zeta^{1-p}_i}{1-p}z_i   -  \frac{\zeta^{2-p}_i}{2-p} \right) + a^R\left(z_i,\phi,p\right)\right\}.
     \end{equation}
     \item \textbf{Gradient Function:}
     
     \begin{align} \label{ofwe.twee.eq.9}
     \nabla^R(\bm{\beta}) &= \frac{1}{\phi}\left( \sum_{i =1}^n t_i\frac{z_i - \zeta_i}{\zeta^{p - 1}_i},\ldots, \sum_{i =1}^n t_i\frac{z_i  - \zeta_i}{\zeta^{p - 1}_i}x_{i,q} \right)^\top
	\\\label{ofwe.twee.eq.10}
	&= \frac{1}{\phi} \bm{X}^\top\bm{D}^R\bm{R}^R;
\end{align} 

The vector \( \bm{R}^R \) and the matrix \( \bm{D}^R \) are defined as follows:

$$\bm{D}^R= \diag\left(t_i \exp\left\{ (2 - p) \mathbf{X}^\top_i \bm{\beta} \right\}  \right)_{i= 1,...,n}, \bm{R}^R= \left(\frac{z_i}{\zeta_i} - 1\right) = \left(\frac{y_i}{\mu_i} - 1\right)_{i= 1,...,n},$$
where $\mu_i = t_i \zeta_i = t_i \exp\left(\mathbf{X}^\top_i \bm{\beta} \right).$
\end{itemize}

\subsection{Link between Offset and Ratio Approaches}\label{ofwe.twee.3}

Building on the Tweedie Invariance Property defined in Proposition \ref{rap.twee.prop}, the equivalence of the two approaches can be demonstrated as follows:

\begin{enumerate}
\item \textbf{The offset approach viewed in the form of the ratio approach}

In the offset approach, \(Y_i|\mathbf{X}^\top_i\) follows a Tweedie distribution with parameters \( (t_i \exp\left( \mathbf{X}^\top_i \bm{\beta} \right), 1, \phi, p) \). We can examine the ratio of the loss cost to the risk exposure \(t_i\), denoted as \(Z_i = \frac{Y_i}{t_i}\), and apply the Tweedie Invariance Property. It can then be demonstrated that \(Z_i|\mathbf{X}^\top_i\) follows a Tweedie distribution with parameters \( (\exp\left( \mathbf{X}^\top_i \bm{\beta} \right), t_i^{2 - p}, \phi, p) \). More formally, we have the following equivalence:

\[
Y|\mathbf{X}^\top_i \sim \text{Tweedie}(t_i \exp\left( \mathbf{X}^\top_i \bm{\beta} \right), 1, \phi, p) 
\equiv 
Z|\mathbf{X}^\top_i \sim \text{Tweedie}(\exp\left( \mathbf{X}^\top_i \bm{\beta} \right), t_i^{2-p}, \phi, p).
\]

Using equation (\ref{MatrixD}) for the calculation of the matrix \( \bm{D} \), we can demonstrate that the two Tweedie distributions introduced earlier share the same expression:

$$\bm{D} = \diag\left( t^{2 - p}_i \exp\left( (2 - p) \mathbf{X}^\top_i \bm{\beta} \right) \right)_{i= 1,\ldots,n} \equiv \bm{D}^O$$

This result provides an additional verification that these two regression approaches are equivalent, leading to identical inference outcomes for the parameter vector \( \bm{\beta} \).\\

\item \textbf{The ratio approach viewed in the form of the offset approach}

Similarly, in the ratio approach, where \( Z_i|\mathbf{X}^\top_i \) is assumed to follow a Tweedie distribution with parameters \( (\exp\left( \mathbf{X}^\top_i \bm{\beta} \right), t_i, \phi, p) \), we can apply the Tweedie Invariance Property to demonstrate that \( Y_i|\mathbf{X}^\top_i \) follows a Tweedie distribution with parameters \( (t_i \exp\left( \mathbf{X}^\top_i \bm{\beta} \right), t_i^{p - 1}, \phi, p) \). This equivalence is expressed as:

\[
Z |\mathbf{X}^\top_i\sim \text{Tweedie}(\exp\left( \mathbf{X}^\top_i \bm{\beta} \right), t_i, \phi, p) 
\equiv 
Y|\mathbf{X}^\top_i \sim \text{Tweedie}(t_i \exp\left( \mathbf{X}^\top_i \bm{\beta} \right), t_i^{p-1}, \phi, p).
\]

For each form of the Tweedie distributions, we also have the following equivalence for the matrix \( \bm{D} \):

\[
\bm{D}= \diag\left( t_i \exp\left( (2 - p) \mathbf{X}^\top_i \bm{\beta} \right) \right)_{i=1,\ldots,n} \equiv \bm{D}^R.
\]

\end{enumerate}

The previous results facilitate a direct comparison of the two approaches. By focusing on the random variable \( Z \), the choice is between the following two distributions:

\[
\text{Tweedie}(\exp\left( \mathbf{X}^\top_i \bm{\beta} \right), t_i^{2-p}, \phi, p) 
\ \ \ \text{\textit{versus}} \ \ \ 
\text{Tweedie}(\exp\left( \mathbf{X}^\top_i \bm{\beta} \right), t_i, \phi, p)
\]

We observe that the choice of approach depends solely on the form of the weights $w_i$, as all other parameters remain unchanged. In other words, selecting between an offset or a ratio approach implies working with normalized loss costs, i.e., the random variable \( Z_i = \frac{Y_i}{t_i} \), and deciding, for observation \(i\), whether to apply a weight of \( w_i^O = t_i^{2-p} \) or a weight of \( w_i^R = t_i \).

\subsection{Weights Analysis}

Since the distinction between the two approaches primarily hinges on the choice of the weight parameter \(w_i\), it is crucial to analyze its impact on modelling loss costs with the Tweedie distribution in greater detail. An initial evident result is the dominance of the weight in the ratio approach compared to that in the offset approach, given by:

$$w_i^O = t_i^{2-p} \ge t_i = w_i^R, \quad \forall p \in ]1, 2[, \, t \in ]0, 1].$$

In Figure \ref{weightsgraph}, we analyze in more detail the differences between the weights used in the offset and ratio approaches. The diagonal dashed line represents the weights in the ratio approach \((w_i^R = t_i)\), while the other coloured lines illustrate the weights \(t_i^{2-p}\) employed in the offset approach for various values of the variance parameter \(p\). 

\begin{figure}[H]
	\begin{center}
		\includegraphics[scale=0.50]{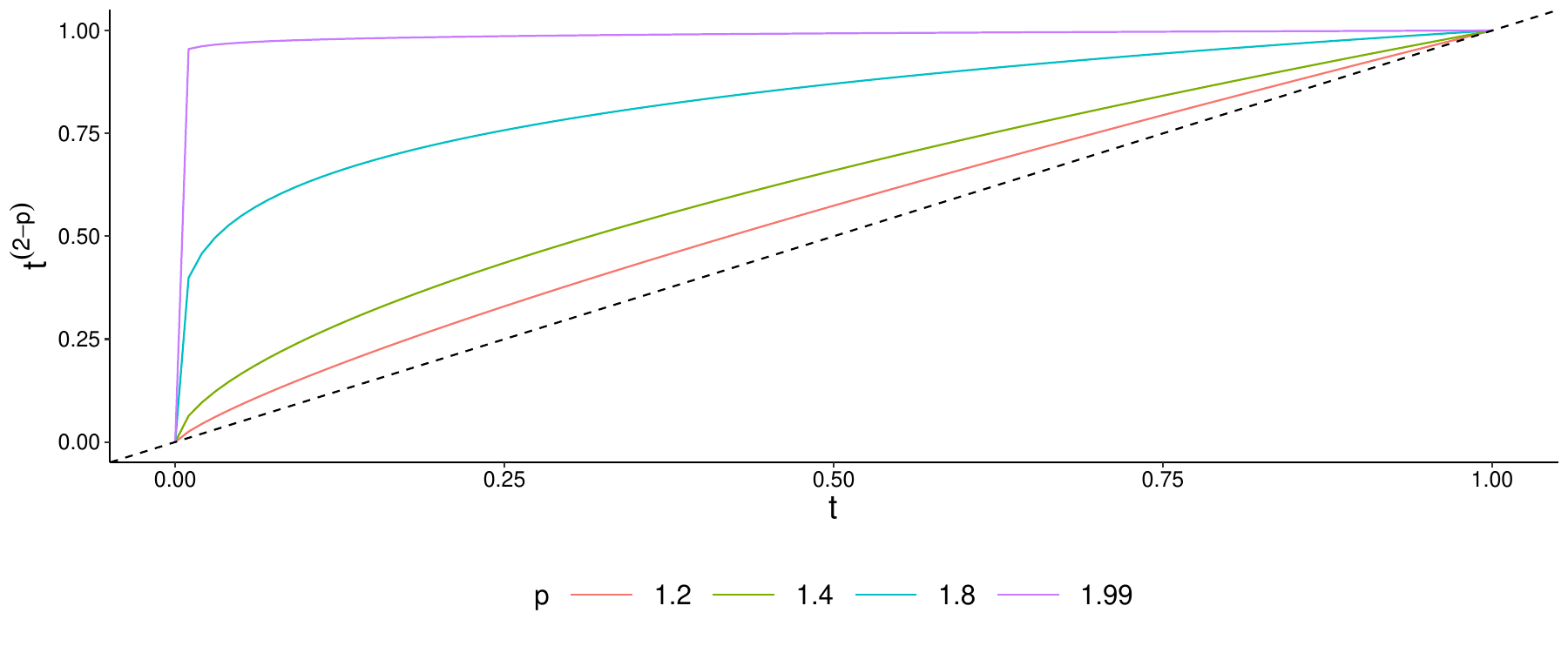}
		\caption{Comparison of the Weight Parameters in the Ratio and Offset Approaches for Different Risk Exposures ($t$)}
		\label{weightsgraph}
	\end{center}
\end{figure}

For a given value of \(p\), we observe that the difference in weights between the offset and ratio approaches diminishes as the risk exposure increases. This indicates that when the contract exposure period is short, the disparity between the weights in both approaches is significant compared to a contract with an exposure period close to one year. In the case where the exposure period equals one year \((t = 1)\), the weights in each approach equal 1, which means that these two regression approaches are equivalent. This equivalence between the methods is also evident in the equality between \(\bm{D}^{R}\) and \(\bm{D}^{O}\) when \(t = 1\):

$$ \bm{D}^{R} = \bm{D}^{O} = \diag\left(\exp\left((2 - p)\mathbf{X}^\top_i \bm{\beta} \right)   
\right)_{i= 1,...,n}.$$

We can also interpret Figure \ref{weightsgraph} by varying the variance parameter \(p\). As \(p\) decreases and approaches 1, the difference between the weights in the offset and ratio approaches becomes negligible. Since the Tweedie distribution converges to the Poisson distribution as \(p \rightarrow 1\), it is noteworthy that this observation aligns with the equivalence of the methods, corresponding to the result obtained regarding claim numbers in Example \ref{exPoisson} presented at the beginning of the paper. Conversely, as \(p\) approaches 2, the disparity in weights between the two approaches becomes more pronounced, indicating that the choice of approach has a more significant impact. Notably, when \(p \rightarrow 2\), the Tweedie distribution converges to a gamma distribution, which has historically been employed to model the severity of claims in actuarial science (see \cite{denuit2007actuarial}). In the context of severity modelling, it is well known that the risk exposure \(t_i\) should not be utilized, implying that the weight should remain constant regardless of the value of the risk exposure \(t_i\). The ratio approach with \(p = 1.999\) tends to demonstrate this characteristic. Thus, overall, these results suggest a preference for the ratio approach, which appears appropriate for any value of \(p \in ]1, 2[\).

\section{Comparison of the Approaches} \label{balance.Twee}

According to the analysis presented in Section~\ref{ofwe.twee.3}, for a fixed value of the variance parameter \(p\), the difference between the weights \(w_i^O\) and \(w_i^R\) associated with the offset and ratio approaches, respectively, depends on the risk exposure \(t_i\). This difference is substantial for contracts with short exposure periods, while it becomes negligible for contracts with exposure periods close to one year. However, a limitation of this analysis is that it does not, by itself, allow us to determine which approach is more appropriate for insurance pricing.

To address this issue, we proceed in two steps. First, we compare the parameter vector estimators derived under the two approaches, relying on the quasi-maximum likelihood framework of \cite{white1982maximum}. Second, using a numerical illustration, we highlight the practical merits of the ratio approach relative to the offset approach in the context of automobile insurance pricing.  To this end, we consider as the response variable the loss cost \(Y\) per unit of risk exposure \(t\), namely $Z = \frac{Y}{t}$, and assume that the true data-generating process of \(Z\) is characterized by a density function \(g\).

For all \(\bm{x} \in \mathbb{R}^{q+1}\) and \(t \in (0,1]\), representing the covariate vector (risk characteristics) and the risk exposure, respectively, the true expected value of \(Z\) is given by

\[
\zeta^{\text{True}} = \int_0^{+\infty} z \, g\!\left(z \mid \bm{\beta}^{\text{True}}, t \right) \, dz = \exp\!\left\{ \bm{x}^\top \bm{\beta}^{\text{True}} \right\}< +\infty,
\]

where \(\bm{\beta}^{\text{True}} \in \mathbb{R}^{q+1}\) denotes the true parameter vector.To clarify our strategy, we also recall the definition of estimator consistency, following \cite{casella2024statistical}. An estimator \(\widehat{\bm{\beta}}\) of \(\bm{\beta}^{\text{True}}\) is said to be consistent (or convergent in probability) if

\[
\widehat{\bm{\beta}} \xrightarrow[]{\mathbb{P}} \bm{\beta}^{\text{True}}
\quad \text{as } n \to +\infty,
\]

that is, if the following condition holds:

\[
\lim_{n \to \infty}
\Pr\!\left( \| \widehat{\bm{\beta}} - \bm{\beta}^{\text{True}} \| > \epsilon \right)
= 0,
\quad \text{for all } \epsilon > 0.
\]

\subsection{Parameter Vector Estimators Comparison} \label{comp.1}

Recall that the estimation of the true parameter vector $\bm{\beta}^{\text{True}}$ under the two approaches is carried out by assuming a Tweedie distribution with different weights in each case:  
the offset approach uses $w = t^{2-p}$, whereas the ratio approach uses $w = t$.  
We denote the corresponding density by $f$, which is defined as

\[
f(z \mid \bm{\beta}, t)
=
\exp\!\left(
\frac{w}{\phi}
\left(
\frac{\exp\!\{(1 - p)\bm{x}^\top \bm{\beta}\}}{1 - p}\, z
-
\frac{\exp\!\{(2 - p)\bm{x}^\top \bm{\beta}\}}{2 - p}
\right)
+ a(z, w, \phi)
\right),
\qquad z \ge 0, \qquad \bm{\beta} \in \mathbb{R}^{q+1}.
\]

We then define the Kullback--Leibler (KL) divergence between the true density $g(\cdot)$ and the density $f(\cdot)$ as

\begin{eqnarray*}
KL(\bm{\beta})
&=&
\mathbb{E}\!\left[
\log \frac{g(z \mid \bm{\beta}^{\text{True}}, t)}
{f(z \mid \bm{\beta}, t)}
\right] \\[0.3em]
&=&
\int_0^{+\infty}
\log\!\big(g(z \mid \bm{\beta}^{\text{True}}, t)\big)\,
g(z \mid \bm{\beta}^{\text{True}}, t)\, dz \\
&&\quad
- \int_0^{+\infty}
\log\!\big(f(z \mid \bm{\beta}, t)\big)\,
g(z \mid \bm{\beta}^{\text{True}}, t)\, dz .
\end{eqnarray*}

\newpage
\begin{proposition} \label{newprop}
According to \cite{white1982maximum}, if the Kullback--Leibler divergence between the true density $g$ and the density $f$ (assumed under either the offset or the ratio approach) exists, then the corresponding estimator of the true parameter vector is consistent under both approaches. More formally, if

\[
KL(\bm{\beta}) < +\infty 
\quad \text{for all } \bm{\beta},
\]

then

\[
\widehat{\bm{\beta}}^{O} \xrightarrow[]{\mathbb{P}} \bm{\beta}^{\text{True}}, \quad \widehat{\bm{\beta}}^{R} \xrightarrow[]{\mathbb{P}} \bm{\beta}^{\text{True}} 
\quad \text{as } n \to +\infty,
\]

where $\bm{\widehat{\beta}}^{O}$ and $\bm{\widehat{\beta}}^{R}$ denote the estimators of 
$\bm{\beta}^{\text{True}}$ obtained from the offset and ratio approaches, respectively.
\end{proposition}

\begin{proof}

The function $KL(\cdot)$ is always nonnegative and finite whenever the probability
measures induced by $g(\cdot \mid \bm{\beta}^{\text{True}}, t)$ and
$f(\cdot \mid \bm{\beta}, t)$ are mutually absolutely continuous (see
\cite{kullback}). Since both $g(\cdot \mid \bm{\beta}^{\text{True}}, t)$ and
$f(\cdot \mid \bm{\beta}, t)$ are densities, this condition is equivalent to
requiring that their associated probability measures be mutually absolutely
continuous, which follows from the Radon--Nikodym theorem.

Suppose that the minimum of $KL(\bm{\beta})$ exists and denote it by
$\bm{\beta}^*$. According to \cite{white1982maximum}, the estimators obtained
under the working density $f$ converge in probability to this minimizer
$\bm{\beta}^*$. Therefore, it suffices to show that this minimizer coincides with
the true parameter vector, that is,
\[
\bm{\beta}^* = \bm{\beta}^{\text{True}}.
\]

We start from the following equivalence:
\[
\arg\min_{\bm{\beta}} KL(\bm{\beta})
=
\arg\max_{\bm{\beta}}
\int_0^{+\infty}
\log\!\big(f(z \mid \bm{\beta}, t)\big)
\, g(z \mid \bm{\beta}^{\text{True}}, t)
\, dz.
\]

For notational convenience, we write
\[
\int_0^{+\infty}
\log\!\big(f(z \mid \bm{\beta}, t)\big)
\, g(z \mid \bm{\beta}^{\text{True}}, t)
\, dz
\;=\; \int_0^{+\infty} \log(f)\, g \, dz.
\]

Using the expression of the Tweedie density, we obtain
\begin{align*}
\int_0^{+\infty} \log(f)\, g \, dz
&= \int_0^{+\infty}
\Bigg[
\frac{w}{\phi}
\Bigg(
\frac{\exp\!\{(1-p)\bm{x}^\top \bm{\beta}\}}{1-p} \, z
-
\frac{\exp\!\{(2-p)\bm{x}^\top \bm{\beta}\}}{2-p}
\Bigg)
+ a(z,w,\phi)
\Bigg]
g(z \mid \bm{\beta}^{\text{True}}, t)
\, dz \\
&\propto
\int_0^{+\infty}
\Bigg(
\frac{\exp\!\{(1-p)\bm{x}^\top \bm{\beta}\}}{1-p} \, z
-
\frac{\exp\!\{(2-p)\bm{x}^\top \bm{\beta}\}}{2-p}
\Bigg)
g(z \mid \bm{\beta}^{\text{True}}, t)
\, dz,
\end{align*}
where the proportionality follows from the fact that $a(z,w,\phi)$ does not
depend on $\bm{\beta}$.

By linearity of the integral, this expression becomes
\begin{align*}
\frac{\exp\!\{(1-p)\bm{x}^\top \bm{\beta}\}}{1-p}
\int_0^{+\infty} z \, g(z \mid \bm{\beta}^{\text{True}}, t) \, dz
-
\frac{\exp\!\{(2-p)\bm{x}^\top \bm{\beta}\}}{2-p}
\int_0^{+\infty} g(z \mid \bm{\beta}^{\text{True}}, t) \, dz.
\end{align*}

Since
\[
\int_0^{+\infty} z \, g(z \mid \bm{\beta}^{\text{True}}, t) \, dz
=
\zeta^{\text{True}}
=
\exp\!\left\{\bm{x}^\top \bm{\beta}^{\text{True}}\right\},
\qquad
\int_0^{+\infty} g(z \mid \bm{\beta}^{\text{True}}, t) \, dz = 1,
\]
the objective function can be written as
\[
K(\bm{\beta})
=
\frac{\exp\!\{(1-p)\bm{x}^\top \bm{\beta} + \bm{x}^\top \bm{\beta}^{\text{True}}\}}{1-p}
-
\frac{\exp\!\{(2-p)\bm{x}^\top \bm{\beta}\}}{2-p}.
\]

The gradient of $K(\bm{\beta})$ with respect to $\bm{\beta}$ is
\[
\nabla_{\bm{\beta}} K(\bm{\beta})
=
\Big(
\exp\!\{(1-p)\bm{x}^\top \bm{\beta} + \bm{x}^\top \bm{\beta}^{\text{True}}\}
-
\exp\!\{(2-p)\bm{x}^\top \bm{\beta}\}
\Big)\bm{x}.
\]

At the minimizer $\bm{\beta}^*$, the first-order condition yields
\[
\nabla_{\bm{\beta}} K(\bm{\beta}^*) = \bm{0}
\quad \Longleftrightarrow \quad
\bm{\beta}^* = \bm{\beta}^{\text{True}},
\]
which completes the proof.
\end{proof}

\subsection{Asymptotic Variance Approximation}

Another important result established by \cite{white1982maximum} is that, under both the offset and ratio approaches, the estimators of the true parameter vector \(\bm{\beta}^{\text{True}}\) are asymptotically Gaussian with a common mean equal to \(\bm{\beta}^{\text{True}}\). Once consistency is established, differences between the two approaches therefore arise solely through their asymptotic dispersion rather than through their limiting value.

From a practical standpoint, differences in asymptotic dispersion translate into differences in the precision and stability of the estimated covariate effects. A smaller asymptotic covariance matrix implies tighter confidence
intervals for the regression coefficients, which facilitates the comparison, selection, and validation of pricing variables, particularly when some effects
are modest but economically relevant. At the same time, reduced dispersion limits the sensitivity of the estimated parameters to random portfolio fluctuations.

In this context, the comparison of asymptotic covariance matrices provides a natural and meaningful criterion for assessing the relative efficiency of the two estimation strategies. Using the asymptotic expansion derived in
Appendix~\ref{appendix2}, we may therefore state the following proposition.

\begin{proposition}\label{proposition.00}
Consider the leading Gaussian terms of the Edgeworth expansion of the asymptotic covariance matrices of 
$\widehat{\bm{\beta}}^{O}$ and $\widehat{\bm{\beta}}^{R}$ associated with the offset and ratio approaches. 
Assume that these leading terms are given by
\begin{align}
\Sigma^{O}&=\phi(\bm{X}^\top\bm{D}^{O}\bm{X})^{-1},\label{asympt.Cov.O} \\
\Sigma^{R}&=\phi(\bm{X}^\top\bm{D}^{R}\bm{X})^{-1}.\label{asympt.Cov.R}
\end{align}
Under the assumptions of Proposition~\ref{newprop} and for risk exposure satisfying \(t\le1\), these leading Gaussian terms satisfy
\[
\Sigma^{R}-\Sigma^{O}\succ0.
\]

\end{proposition}

\begin{proof}
First, consider $n$ observations of the loss cost per unit of exposure $Z_1,\ldots,Z_n$, with $Z_i\sim\twee(\zeta_i,w_i,\phi,p)$, $w_i\in\{t_i,t_i^{2-p}\}$, and $\zeta_i=\exp\{\bm{X}_i^\top\bm{\beta}\}$. According to \cite{white1982maximum}, the asymptotic covariance matrix of the estimator is given by $\Sigma=A^{-1}BA^{-1}$, where $A=-I_n(\bm{\beta})$ and
$B=\mathbb{E}[\nabla(\bm{\beta})\nabla(\bm{\beta})^\top]$, with

\[
\nabla(\bm{\beta})=\frac{1}{\phi}\Big(\sum_{i=1}^nw_i\frac{Z_i-\zeta_i}{\zeta_i^{p-1}}x_{i,0},\ldots,
\sum_{i=1}^nw_i\frac{Z_i-\zeta_i}{\zeta_i^{p-1}}x_{i,q}\Big)^\top.
\]

As shown in Appendix~\ref{appendix2}, $B=\phi^{-1}\bm{X}^\top\bm{D}\bm{X}=-A$, yielding

\[
\Sigma=A^{-1}BA^{-1}=-A^{-1}=\phi(\bm{X}^\top\bm{D}\bm{X})^{-1}.
\]

We now compare the two covariance matrices. As noted in Section~\ref{rap.twee.4}, $\Sigma^{O},\Sigma^{R}\in\mathcal{P}_{q+1}$. Define
$M=\Sigma^{R}-\Sigma^{O}$. To prove that $M$ is positive definite, it suffices, by the results of Equation~(\ref{PosMt.pro.2}), to show that $\bm{D}^{O}-\bm{D}^{R}$ is positive definite. From Figure~\ref{weightsgraph}, we have

\[
\forall v\in\mathbb{R}^n,\quad
v^\top(\bm{D}^{O}-\bm{D}^{R})v
=\sum_{i=1}^n(t_i^{2-p}-t_i)\zeta_i^{2-p}v_i^2>0.
\]

Moreover, since $\bm{X}$ has full rank $q+1$, equation~(\ref{PosMt.pro.3}) implies

\[
\bm{X}^\top(\bm{D}^{O}-\bm{D}^{R})\bm{X}
=\bm{X}^\top\bm{D}^{O}\bm{X}-\bm{X}^\top\bm{D}^{R}\bm{X}\succ0,
\]

It then follows from Equation~\eqref{PosMt.pro.2} that
\( M = \Sigma^{R} - \Sigma^{O} \) is positive definite.

\end{proof}

Taken together, the previous propositions show that each approach can be implemented independently and yields a consistent (i.e., convergent) estimator of the true parameter vector $\bm{\beta}^{\text{True}}$.  However, based solely on the result stated in Proposition~\ref{proposition.00}, one might be led to conclude that the estimator obtained under the offset approach is asymptotically more efficient than the one obtained under the ratio approach.

It is important to emphasize that Proposition~\ref{proposition.00} is formulated exclusively at the level of the leading Gaussian term of the Edgeworth expansion and therefore concerns only first-order asymptotic quantities. More specifically, the matrices $\Sigma^{O}$ and $\Sigma^{R}$ represent the leading $O(1/n)$ terms in the covariance expansions of $\widehat{\bm{\beta}}^{O}$ and $\widehat{\bm{\beta}}^{R}$. 
Consequently, the difference $\Sigma^{R}-\Sigma^{O}$ identified in Proposition~\ref{proposition.00} is itself a first-order quantity of order $O(1/n)$. The apparent superiority of the offset approach therefore relies on a comparison between leading terms only and is meaningful insofar as the full asymptotic covariance matrices of $\widehat{\bm{\beta}}^{O}$ and $\widehat{\bm{\beta}}^{R}$ are sufficiently well approximated by their leading Gaussian components $\Sigma^{O}$ and $\Sigma^{R}$, as defined in equations~(\ref{asympt.Cov.O}) and~(\ref{asympt.Cov.R}). 
Since higher-order contributions in the Edgeworth expansion are neglected, the resulting approximation error is of smaller order asymptotically, but may nevertheless be of the same order of magnitude as the first-order difference $\Sigma^{R}-\Sigma^{O}$ in finite samples.  As a result, it is likely inappropriate to draw definitive conclusions regarding the relative performance of the two approaches based solely on Proposition~\ref{proposition.00}. 
A natural extension would therefore consist in deriving and comparing higher-order terms in the Edgeworth expansion. 
While such an analysis would be of clear theoretical interest, it lies beyond the scope of the present paper.

\subsubsection{Simulations Study}  
Instead, we focus on assessing the accuracy of the asymptotic approximation and on comparing the two methods through a simulation study. The simulation study is designed solely to assess the empirical relevance of Proposition \ref{proposition.00}, which provides a first-order asymptotic comparison between the offset and ratio approaches based on the leading Gaussian term of the Edgeworth expansion of the estimators’ covariance matrices. To this end, independent samples are generated under a data-generating process that is fully consistent with the modeling assumptions underlying the theoretical result. For each scenario, defined by a combination of sample size
\(n\) and Tweedie power parameter \(p\), repeated Monte Carlo replications are performed, and both the offset and ratio estimators are computed using the same Tweedie specification as that used for data generation.

For each approach, the resulting collection of parameter estimates is used to construct empirical covariance matrices, which provide finite-sample
approximations of the true sampling covariances of the estimators. The empirical counterpart of the theoretical comparison is then obtained by examining the
difference between these covariance matrices and by evaluating its spectral properties, in particular through the sign of its smallest eigenvalue. Across all scenarios considered, the empirical results are consistent with the
ordering predicted by Proposition \ref{proposition.00}. In particular, the empirical difference \(\widehat{\Sigma}^{R} - \widehat{\Sigma}^{O}\) is found to be positive definite in all tested configurations.

From the standpoint of parameter estimation accuracy, the offset approach may be
preferable to the ratio approach in finite samples, as it consistently yields
smaller empirical variances for the estimator of \(\beta\). This empirical
finding is in line with the asymptotic ordering established in Proposition
\ref{proposition.00}.

\subsection{Sum of Individual Observed Gaps Comparison}\label{eqfinance}

In the context of automobile insurance pricing, one of the primary objectives is not only parameter estimation but also the accurate determination of premiums. 
From this perspective, we instead focus on proximity to financial balance as a key criterion for evaluating the different approaches. 
As discussed in \cite{denuit2024testing}, financial balance broadly refers to the situation in which the sum of observed total losses coincides with the sum of premiums computed under a given pricing method. 
Since this balance is not exactly achieved under either approach due to the use of a logarithmic link function in premium modeling, financial balance naturally emerges as an important criterion to consider when selecting between the two approaches.

We consider $n$ independent contracts with observed loss costs $y_i=t_i z_i$ for $i=1,\ldots,n$, where $z_i$ denotes the loss cost per unit of risk exposure $t_i$. We define the \textit{individual observed gap} as the difference between the observed loss cost and the estimated premium. For contract $i$, these gaps are denoted by $\Delta_i^{O}$ and $\Delta_i^{R}$ under the offset and ratio approaches, respectively:

\[
\Delta_i^{O}=t_i\bigl(z_i-\widehat{\zeta}_i^{O}\bigr),
\qquad
\Delta_i^{R}=t_i\bigl(z_i-\widehat{\zeta}_i^{R}\bigr),
\]

where $\widehat{\zeta}_i^{O}$ and $\widehat{\zeta}_i^{R}$ are the estimated premiums obtained from the offset and ratio approaches. These quantities can be interpreted as realizations of the corresponding premium estimators. Our
quantities of interest are therefore the aggregated gaps

\begin{equation}
\sum_{i=1}^n \Delta_i^{O}
\quad \text{and} \quad
\sum_{i=1}^n \Delta_i^{R}.
\label{eqdiffgaps}
\end{equation}

Although these sums are generally not equal to zero—reflecting the use of a logarithmic link function rather than the canonical link of the Tweedie model in premium modeling—we show that, under the ratio approach, the sum of individual observed gaps tends to be closer to zero than under the offset approach.

To analyze this criterion, we rely on the same simulation framework as in the previous section and examine, for each replication with $p=1.5$, the ratio between the total predicted premiums and the total observed losses under both approaches. This comparison is illustrated in Figure~\ref{RatioMonteCarlo}. Across all considered sample sizes, the distribution associated with the ratio approach is markedly more concentrated around the value 1—corresponding to a situation close to financial balance—than the distribution obtained under the offset approach. In contrast, the offset approach exhibits substantially greater dispersion, reflecting higher variability around the global balance condition, even as the sample size increases.

A complementary perspective consists in identifying, for each simulation, which approach yields a total predicted premium closer to the observed total losses. According to this criterion, and for all scenarios under consideration, the ratio approach is closer to financial balance in approximately 87\% of the replications. This result further supports the view that, from the standpoint of aggregate financial balance, the ratio approach displays a more stable and robust behavior than the offset approach in finite samples.

\begin{figure}[H]
	\begin{center}
		\includegraphics[scale=0.45]{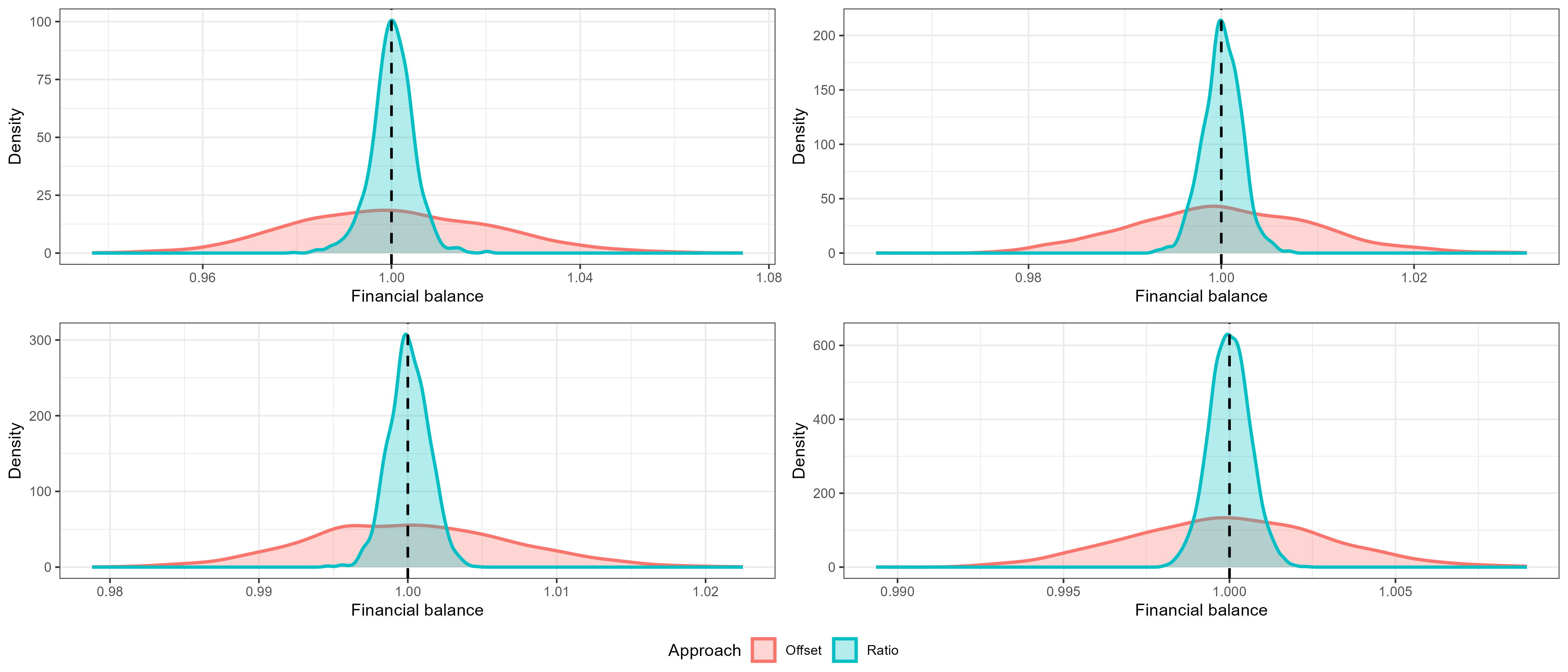}
		\caption{Financial balance density for $n = 1{,}000$, $5{,}000$, $20{,}000$, and $50{,}000$ under the offset and ratio approaches.}
		\label{RatioMonteCarlo}
	\end{center}
\end{figure}

Nevertheless, rather than relying solely on simulation results, we argue that additional insight into this financial balance criterion can be gained by examining the parameter estimation equations themselves. To this end, we first consider a homogeneous setting in which the sum of individual observed gaps is exactly zero under the ratio approach. We then extend the analysis to a heterogeneous setting in order to compare the behavior of this aggregate gap under the ratio and offset approaches.

\subsubsection{Homogeneous Portfolio} \label{scorehyp}

First, we will assume the analysis of a homogeneous portfolio for which no explanatory variable is used in the mean parameter of the two approaches under study. More formally, we assume that the loss costs \( Z_1, \ldots, Z_n \) are independent and have a common mean parameter \( \zeta \). We also assume that \( \zeta_i = \zeta \) for all contracts \( i = 1, \ldots, n \). This allows us to express the log-likelihood function for \( \zeta \) as follows:

\begin{align}
    \ell(\zeta) = \frac{1}{\phi} \left( \frac{\zeta^{1 - p}}{1 - p} \sum_{i=1}^n w_iz_i  -  \frac{\zeta^{2 - p}}{2 - p} \sum_{i=1}^n w_i\right) + \sum_{i=1}^n a\left(z_i,w_i,\phi,p\right).
\end{align}

The first derivative of this log-likelihood function is given by:

\begin{align}
    \ell^{'}(\zeta) = \frac{1}{\phi \zeta^p} \left(  \sum_{i=1}^n w_iz_i  -  \zeta \sum_{i=1}^n w_i\right).
\end{align}

Thus, the quasi-maximum likelihood estimator (QMLE) of $\zeta$, denoted by \( \widehat{\zeta} \), is obtained as follows: \( \widehat{\zeta} = \frac{\sum_{i=1}^n w_i z_i}{\sum_{i=1}^n w_i} \).

\begin{enumerate}
    \item \textbf{Offset Approach:} Under the offset approach, the QMLE of the homogeneous risk parameter $\zeta$ is given by
\begin{equation}\label{hom.off}
\widehat{\zeta}^{O}
=
\frac{\sum_{i=1}^n t_i^{2-p} z_i}{\sum_{i=1}^n t_i^{2-p}}.
\end{equation}
The corresponding aggregate estimated premium is therefore
\begin{equation}
\sum_{i=1}^n t_i \widehat{\zeta}^{O}
=
\left(\sum_{i=1}^n t_i\right)
\frac{\sum_{i=1}^n t_i^{2-p} z_i}{\sum_{i=1}^n t_i^{2-p}}.
\end{equation}

Exact financial balance, i.e.
\[
\sum_{i=1}^n t_i \widehat{\zeta}^{O}
=
\sum_{i=1}^n t_i z_i,
\]
holds if and only if the $t_i^{2-p}$-weighted average of the observed loss costs $z_i$ coincides with their $t_i$-weighted average. 
This condition is automatically satisfied in degenerate cases such as $p=1$, or when all exposures $t_i$ are equal, in which case the two weighting schemes coincide. However, under the general setting considered in this paper, where $1<p<2$ and exposures satisfy $t_i\le 1$ with heterogeneous durations, there is no structural mechanism in the offset model that enforces this identity.
As a result, aggregate financial balance is not guaranteed under the offset approach and may fail in general, although it can still occur for particular realizations of the portfolio.
Consequently, the aggregate deviation
\[
\sum_{i=1}^n \Delta_i^{O}
=
\sum_{i=1}^n t_i \widehat{\zeta}^{O}
-
\sum_{i=1}^n t_i z_i
\]
is not constrained to vanish.

    \item \textbf{Ratio Approach:} In the ratio approach, where \( Z_i \) also serves as the response variable, the weight for contract \( i \) is defined as \( w_i = t_i \). The QMLE for \( \zeta \), denoted by \( \widehat{\zeta}^{R} \), is calculated as:

\begin{equation} \label{hom.weig}
    \widehat{\zeta}^{R} = \frac{\sum_{i=1}^n t_i z_i}{\sum_{i=1}^n t_i}.
\end{equation}

 In this case, the aggregate observed loss costs matches the sum of the expected premiums, as shown by:

\begin{align*}
    \sum_{i=1}^n t_i \widehat{\zeta}^{R} &=    \sum_{i=1}^n t_i \left(\frac{\sum_{i=1}^n t_i z_i}{\sum_{i=1}^n t_i} \right)\\
    &= \frac{\sum_{i=1}^n t_i z_i}{\sum_{i=1}^n t_i}\sum_{i=1}^n t_i\\
    &= \sum_{i=1}^n t_i z_i.
\end{align*}

So, we conclude that \(\sum_{i=1}^n\Delta_i^R = 0\). This result suggests that the ratio approach inherently ensures that the sum of estimated premiums aligns with the aggregate observed loss costs, making it advantageous for achieving balance at the portfolio level.
\end{enumerate}

\subsubsection{Heterogeneous Portfolio} \label{nonscorehyp}

The homogeneous framework can be extended to account for heterogeneity by incorporating the covariate vector $\mathbf{X}_i$ in the premium model for each contract $i$. This allows for contract-specific premium estimation and captures variation across contracts. In contrast to the homogeneous portfolio, where a closed-form expression for the QMLE of $\bm{\beta}$ was available, the presence of covariates implies that the QMLE of the parameter vector $\bm{\beta}$ no longer admits a closed-form solution. This is due to its dependence on $\mathbf{X}_i$. Consequently, the QMLE of $\bm{\beta}$ is obtained using the IRLS algorithm (see Section~\ref{rap.twee.4}) for both the offset and ratio approaches, as discussed in the previous sections.

To better understand the gap between the aggregate observed loss costs and the sum of estimated premiums in each approach, we initialize the IRLS algorithm with values of \( \bm{\beta} \) based on the homogeneous portfolio assumption. This initialization promotes consistency in comparing the offset and ratio approaches and clarifies how incorporating  \( \mathbf{X}_i\) affects the alignment of estimated premiums with aggregate loss costs.

\begin{enumerate}
    \item \textbf{Offset Approach:} We initialize the IRLS algorithm with:

    $$\beta_0 = \log\left(\widehat{\zeta}^{O}\right),~\text{and} ~ \beta_j = 0, ~\text{for}~ j=1,\ldots,q.$$

    By denoting the initial value by \( \bm{\beta}^{O}_{(0)} \), the estimate of the parameter vector at iteration \( K \), $\bm{\beta}^{O}_{(K)}$,  is obtained as follows:

    $$ \bm{\beta}^{O}_{(K)} = K \bm{\beta}^{O}_{(0)} + \underbrace{ \sum_{k=1}^{K - 1} \left(I^{O}_n\left(\bm{\beta}^{O}_{(k)}\right)\right)^{-1} \nabla^{O}\left(\bm{\beta}^{O}_{(k)} \right)}_{\delta^O_{K - 1}} = K \bm{\beta}^{O}_{(0)} + \delta^O_{K - 1},$$

     where \( I^{O}_n(\cdot) \) is the Fisher information matrix for the offset approach. Therefore, we calculate the estimate premium at iteration \( K \), denoted by  \(\widehat{\zeta}^O_{i(K)} \), as follows:

     $$\widehat{\zeta}^O_{i(K)} =  \left(\widehat{\zeta}^{O}\right)^K \exp\left\{  \mathbf{X}^{T}_i \delta^O_{K - 1}\right\},$$

where $\widehat{\zeta}^{O}$, defined by equation (\ref{hom.off}), represents the annual premium for all insureds in the offset approach for the homogeneous portfolio. We can conclude from this last equation that the sum of the estimated premiums at iteration \( K \), denoted as \( \sum_{i=1}^{n} t_i\widehat{\zeta}^O_{i(K)} \), does not necessarily equal the sum of the observed loss costs, \( \sum_{i=1}^{n} t_i z_i \). This discrepancy is typical in the offset approach, given its weighted nature and the way the IRLS updates progress over iterations.

     \item \textbf{Ratio Approach:}  We initialize the IRLS algorithm with:

    $$\beta_0 = \log\left(\widehat{\zeta}^{R}\right),~\text{and} ~ \beta_j = 0, ~\text{for}~ j=1,\ldots,q.$$

    By denoting the initial value by \( \bm{\beta}^{R}_{(0)} \), the estimate of the parameter vector at iteration \( K \), $\bm{\beta}^{R}_{(K)}$,  is obtained as follows:

    $$ \bm{\beta}^{R}_{(K)} = K \bm{\beta}^{R}_{(0)} + \underbrace{ \sum_{k=1}^{K - 1} \left(I^{R}_n\left(\bm{\beta}^{R}_{(k)}\right)\right)^{-1} \nabla^{R}\left(\bm{\beta}^{R}_{(k)} \right)}_{\delta^R_{K - 1}} = K \bm{\beta}^{R}_{(0)} + \delta^R_{K - 1},$$

     where \( I^{R}_n(\cdot) \) is the Fisher information matrix for the ratio approach. Therefore, we calculate the estimate premium at iteration \( K \) as follows:

     $$\widehat{\zeta}^R_{i(K)} = \left( \widehat{\zeta}^{R}\right)^K \exp\left\{ \mathbf{X}^{T}_i \delta^R_{K - 1}\right\},$$

where $\widehat{\zeta}^{R}$, defined by equation (\ref{hom.weig}), represents the annual premium for all insureds in the ratio approach for the homogeneous portfolio. This allows us to calculate the sum of the estimated premiums at iteration \( K \), \(t_i \widehat{\zeta}^R_{i(K)}\), leading to the following equation:  

\begin{eqnarray*}
\sum_{i=1}^{n} t_i \widehat{\zeta}^R_{i(K)} &=& \left(\sum_{i=1}^{n} y_i\right) \times \left( \frac{\left( \widehat{\zeta}^{R}\right)^{K - 1}}{\sum_{i=1}^{n} t_i} \sum_{i=1}^{n} t_i\exp\left\{\mathbf{X}^{T}_i \delta^R_{K - 1}\right\}\right) \\
&=& \sum_{i=1}^{n} y_i \times 
\epsilon_{(K-1)},
\end{eqnarray*}

highlighting the connection between the total premiums and the aggregate claims \(\sum_{i=1}^{n} y_i\).  That means that unless $\epsilon_{(K-1)}$ equals 1, there is no financial equilibrium for this approach. While the sum of estimated premiums may not exactly match the observed aggregate loss costs, this formulation allows for an analysis of $\epsilon_{(K-1)}$, the gap between these two quantities. \\
\end{enumerate}

\subsubsection{Numerical Illustration} \label{simulation}  

To better understand the implications of the parameter estimation equations underlying the two approaches—particularly in the heterogeneous case, where no explicit analytical solution is available—we propose a simple numerical illustration. 
The objective of this illustration is not to assess the statistical performance of the estimators, but rather to highlight the structural differences between the estimation equations and to clarify how these differences translate into distinct behaviors with respect to the financial balance criterion.

The setup is as follows:
\begin{enumerate}
    \item \textbf{Portfolio setup.} 
    In the homogeneous case, we consider a portfolio of \( n = 100 \) independent contracts without covariates. 
    In the heterogeneous case, two risk factors are incorporated into the mean specification of each contract. These covariates are generated independently from binomial distributions:
    \begin{itemize}
        \item \( x_1 \sim \text{Binomial}(100, 0.75) \),
        \item \( x_2 \sim \text{Binomial}(100, 0.15) \).
    \end{itemize}
    This specification induces heterogeneity in the portfolio while remaining sufficiently simple to allow a transparent interpretation of the estimation equations.

    \item \textbf{Risk exposure.} 
    Exposure levels \( t_i \) are drawn from a uniform distribution on the interval
    \[
        \left[\tfrac{30}{365}, \tfrac{335}{365}\right],
    \]
    and then sorted in increasing order to assign contract indices \( i = 1, \ldots, n \). 
    This ordering reflects a common actuarial situation in which higher indexed contracts correspond to larger exposure levels.

    \item \textbf{Loss cost.} 
    Rather than assuming a stochastic loss-generating mechanism, we consider a deterministic specification given by
    \[
        y_i = i, \qquad i = 1, \ldots, n.
    \]
    This choice is deliberate and aims to isolate the effect of the estimation equations themselves. By eliminating random noise, the resulting comparison focuses exclusively on the structural differences between the offset and ratio approaches and on their implications for aggregate financial balance.
\end{enumerate}

\begin{figure}[H]
	\begin{center}
		\includegraphics[scale=0.45]{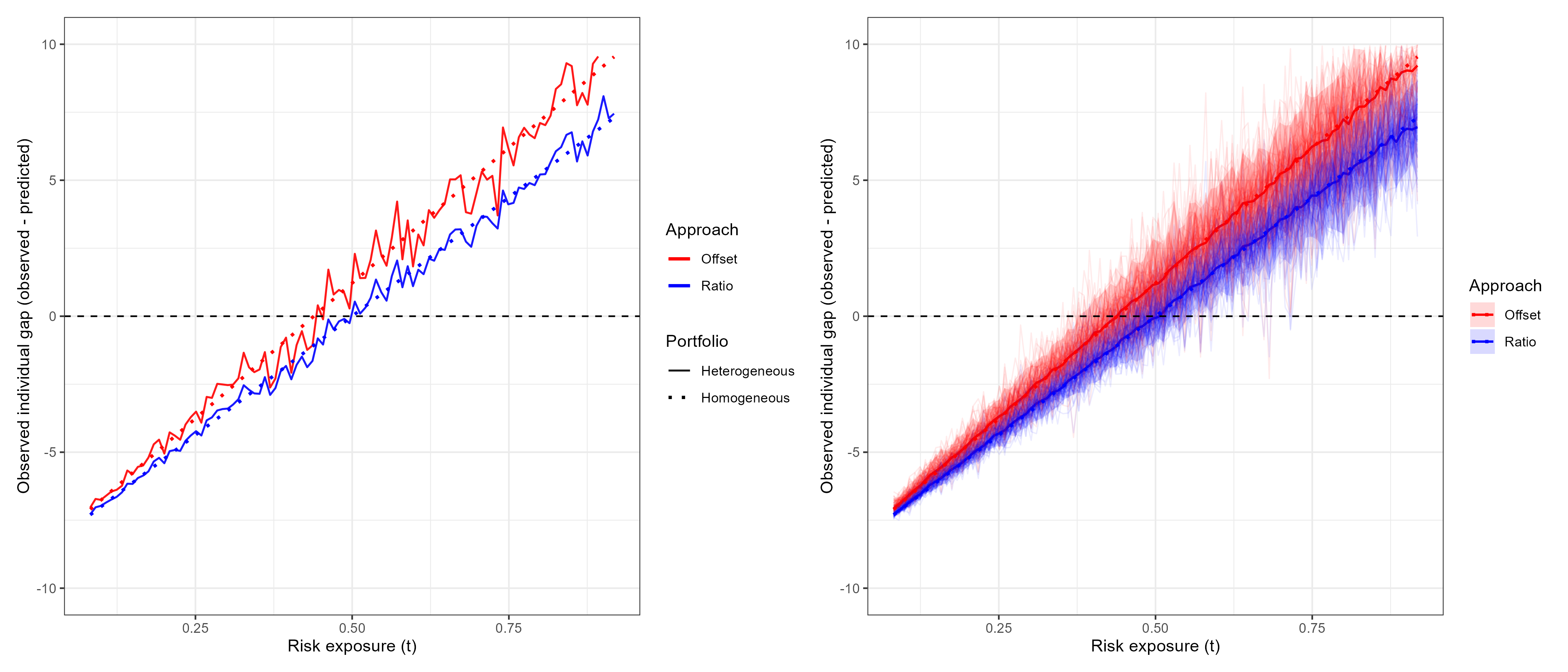}
\caption{Individual observed gaps: single heterogeneous realization (left) and across multiple heterogeneous realizations (right).}
		\label{simul}
	\end{center}
\end{figure}

We then apply the corresponding parameter estimation equations to this portfolio. 
For a single realization, the individual observed gaps are displayed in the left panel of Figure~\ref{simul}. This panel allows a direct comparison across portfolio types (homogeneous versus heterogeneous) and estimation approaches (offset versus ratio). 
Several qualitative insights can be drawn from this figure:

\begin{itemize}
\item \textbf{Ratio approach (solid and dashed blue curves).}

Recall that, for a homogeneous portfolio, the ratio approach satisfies the exact aggregate balance condition $\sum_{i=1}^n \Delta_i^{R} = 0$, as established in Section~\ref{scorehyp}. Although this aggregate quantity is not explicitly reported in the figure, the dashed blue curve corresponds to the individual observed gaps obtained under the homogeneous assumption and therefore sums to zero by construction.

This homogeneous benchmark provides a natural reference for interpreting the solid blue curve, which represents the individual observed gaps obtained under the heterogeneous specification. The two curves are remarkably close over the entire range of exposure levels, indicating that the introduction of covariates does not materially alter the aggregate balance property of the ratio approach. In particular, this visual proximity suggests that the sum \( \sum_{i=1}^n \Delta_i^{R} \) in the heterogeneous case remains close to zero.

At the same time, the small local deviations and step-like variations observed along the solid blue curve reflect the fact that exact financial balance cannot generally be achieved in the heterogeneous setting, except in accidental cases where the sum of these local variations happens to cancel out. These discrepancies arise from the contract-specific nature of the estimated premiums and from the iterative nature of the estimation procedure. Nevertheless, the overall behavior remains strongly anchored to the homogeneous equilibrium, highlighting the robustness of the ratio approach with respect to portfolio heterogeneity.

\item \textbf{Offset approach (solid and dashed red curves).}

As established in Section~\ref{scorehyp}, the offset approach does not satisfy an exact aggregate balance condition, that is, $\sum_{i=1}^n \Delta_i^{O} \neq 0.$  This feature is directly visible in Figure~\ref{simul}, where both red curves display a systematic displacement relative to the ratio approach, which serves as a natural benchmark since its aggregate gap is zero by construction.

In the homogeneous case (dashed red curve), this displacement is already apparent, as the offset curve lies uniformly above its ratio counterpart, indicating a departure from financial equilibrium. In the heterogeneous case (solid red curve), the offset curve remains centered around its homogeneous counterpart, suggesting that the imbalance persists even in the presence of additional local variability arising from contract-specific covariates.
\end{itemize}

To further assess the impact of parameter estimation on financial balance, we consider a Monte Carlo experiment based on 1{,}000 heterogeneous portfolios generated from varying realizations of \(x_1\) and \(x_2\). The corresponding results are reported in the right panel of Figure~\ref{simul}. In this panel, all 1{,}000 individual gap curves are displayed simultaneously as semi-transparent lines, providing a direct visualization of the outcomes associated with heterogeneous portfolios. To summarize this variability, a pointwise 95\% envelope is added in the form of a shaded ribbon, while the bold solid curve represents the median gap across replications. We can see that the individual gap curves obtained across replications are strongly concentrated around their respective homogeneous benchmarks for both approaches. This is reflected by the close proximity of the median heterogeneous curve to the corresponding homogeneous curve shown in the left panel, as well as by the relatively narrow width of the 95\% ribbons over the entire range of risk exposure. These patterns indicate that heterogeneity primarily introduces dispersion around the homogeneous structure, rather than inducing systematic shifts in the gap curves themselves.

An important implication of this concentration phenomenon is that the aggregate imbalance observed under the homogeneous specification provides a meaningful first-order diagnostic for the heterogeneous case. Since heterogeneity mainly adds random fluctuations around the homogeneous benchmark, a substantial lack of financial balance at the homogeneous level is unlikely to be systematically corrected once heterogeneity is introduced. Conversely, when the homogeneous imbalance is small, near-balance under heterogeneous modeling becomes more plausible, albeit still not structurally enforced. In this respect, the two approaches differ fundamentally. Exact aggregate balance holds by construction for the homogeneous ratio approach, whereas no such property holds for the homogeneous offset approach. As a result, the concentration of heterogeneous gap curves around their homogeneous counterparts implies that, in heterogeneous portfolios, the ratio approach is typically closer to financial equilibrium than the offset approach.

Finally, we note that the offset approach may occasionally exhibit a better financial balance than the ratio approach for specific realizations, as observed in the simulation results reported in Figure~\ref{RatioMonteCarlo}, where it is selected in approximately 13\% of the replications. As illustrated in Figure~\ref{simul}, such outcomes can partly be attributed to more pronounced random fluctuations, whose aggregate effect may incidentally compensate for the residual imbalance. However, these situations are also more likely to occur in scenarios where the homogeneous offset specification itself happens to be close to financial balance. In such cases, the additional dispersion introduced by heterogeneity may suffice to produce a near-balanced aggregate outcome under the offset approach. Since this compensation mechanism is not structural but contingent on the particular data realization and on the initial homogeneous imbalance, the ratio approach generally provides a more reliable and robust financial balance across heterogeneous portfolios.

\section{Empirical Illustration}\label{ap.num}

To illustrate the offset and ratio approaches with real insurance data, we utilize a non-random sample from a car insurance database maintained by a major Canadian insurer, spanning 13 consecutive years. The dataset pertains to the province of Ontario and includes over 2 million observations, with each observation representing an annual insurance contract for an individual vehicle. For each contract, the dataset includes key identifiers such as policy number, vehicle ID, and contract start and end dates. The dataset also contains detailed information about the insured individuals and their vehicles, as well as the number and cost of claims associated with each contract. The database is segmented by coverage type, providing insights into third-party liability, collision, and comprehensive claims.

For the purposes of this study, we focus on a single coverage type: collision coverage, which addresses property damage protection for at-fault accidents. To ensure the independence of contracts, we modified the original dataset to include only one vehicle contract per policy. In this revised dataset, each policy corresponds to a unique contract, reducing the dataset to approximately 161,390 contracts.

\subsection{Description of Data}\label{Datadescrip}

\subsubsection{Description of Contracts}\label{Contractdesc}

A standard car insurance contract typically carries a one-year risk exposure; however, cancellations can occur within this period. To offset the administrative costs associated with such cancellations, the insured is subject to a financial penalty for breaching the contract. Several scenarios can lead to mid-term cancellations:

\begin{enumerate}
    \item The insured sells the vehicle and no longer requires insurance coverage.
    \item The insured opts to switch to another insurer, potentially attracted by a lower cost despite the cancellation penalty.
    \item The insured experiences a significant loss (e.g., theft or accident) and subsequently changes vehicles, rendering the original insurance coverage unnecessary.
\end{enumerate}

In cases \#1 and \#3, it could be argued that the situation constitutes a vehicle change rather than an outright cancellation of insurance. In practice, insurers often do not impose penalties for such changes, provided the insured continues coverage with the same insurer for the new vehicle. However, a vehicle change frequently enables the insured to explore options with other insurers, which can quickly lead to scenario \#2.

\subsubsection{Summary of Data}\label{Contractrisk}

In Table \ref{ratio.losscost}, we present a summary of the data available for the numerical application. It shows that the majority of the observations in the database come from insured individuals covered for a full year. However, a significant portion of the contracts observed in the database corresponds to mid-term cancellations. It can be observed that the average exposure of contracts without cancellations is 1, while contracts with mid-term cancellations have an average exposure of approximately 0.5. The table also indicates the average risk exposure for each group. The last column presents what we refer to as the \textbf{Loss Cost Reference}, which is derived by dividing the average loss cost of each group by the portfolio’s average \footnote{The use of the \textit{Loss Cost Reference} helps prevent excessive disclosure of confidential information from the insurer who provided us with the data.}. A notable difference emerges between the two groups, indicating that insureds with mid-term cancellations tend to have higher insurance damages than those in the other group.

 \begin{table}[H]
	\centering
	\begin{tabular}{cccc}
		\hline
		& Proportion   & Average   Risk    & Loss Cost \\
		Group & of contracts &  Exposure & Reference\\ 
		\hline
                Full exposure & 64\% & 1.00 & 0.63\\ 
        Mid-term cancellation & 36\% & 0.51 & 2.45\\
		\hline
	\end{tabular}
	\caption{Descriptive statistics by group of contracts}
	\label{ratio.losscost}
\end{table}

Figure \ref{truncdesc} presents a detailed view of the risk exposure distribution for insured individuals with mid-term cancellations, while the risk exposure for other insured individuals remains constant at 1. The distribution shows that risk exposure is fairly uniform throughout the year; however, there is still considerable variation.

\begin{figure}[H]
	\begin{center}
		\includegraphics[width=0.7\textwidth]{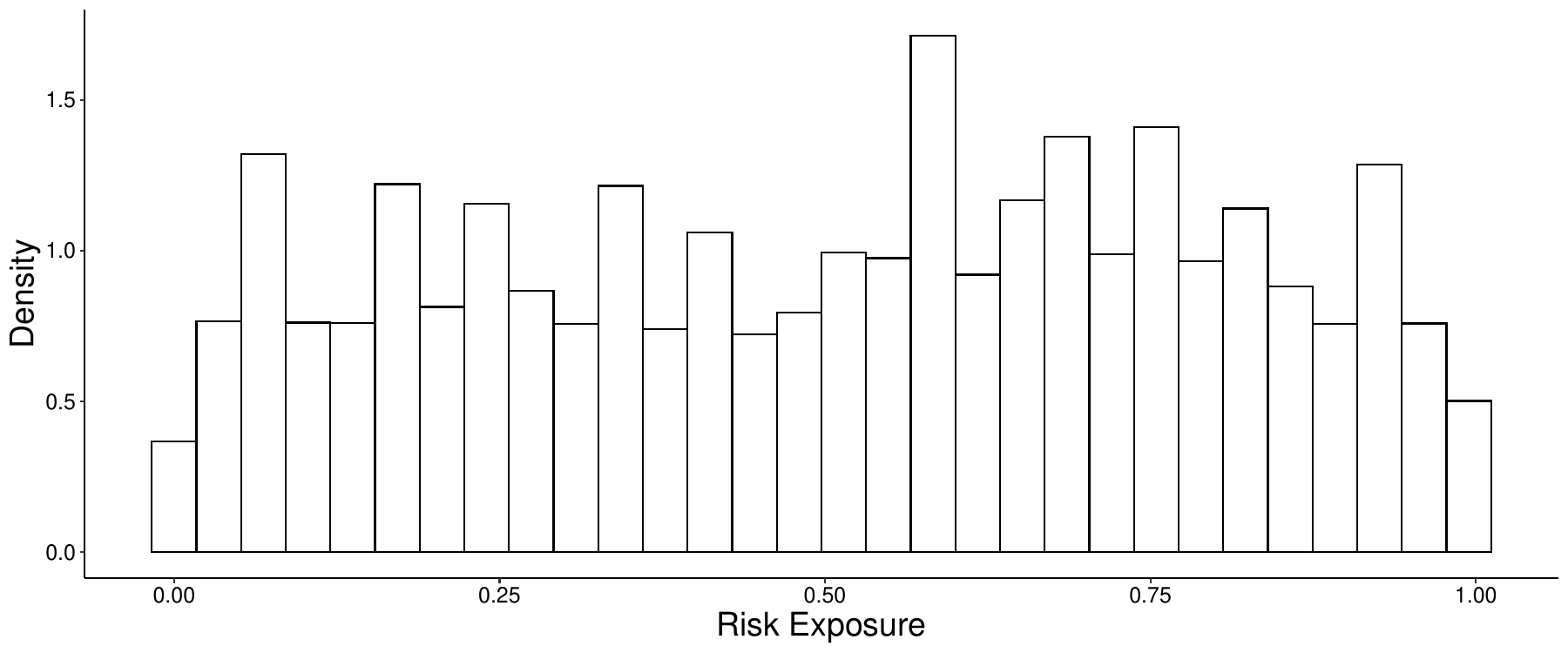}
     \caption{Distribution of risk exposures }
    \label{truncdesc}
	\end{center}
\end{figure}

To further examine the relationship between risk exposure and loss costs, we grouped contracts by the total number of coverage months, from 1 to 11. For each coverage duration, we calculated the \textit{loss cost reference} and plotted these values against the average risk exposure, as shown in Figure \ref{statdesbygroup}. Insured individuals with full exposure (coverage for 12 months) do not appear in the groups for months 1 to 11; thus, only their average value is displayed in Figure \ref{statdesbygroup} as a dashed line.

\begin{figure}[H]
	\begin{center}
		\includegraphics[width=0.7\textwidth]{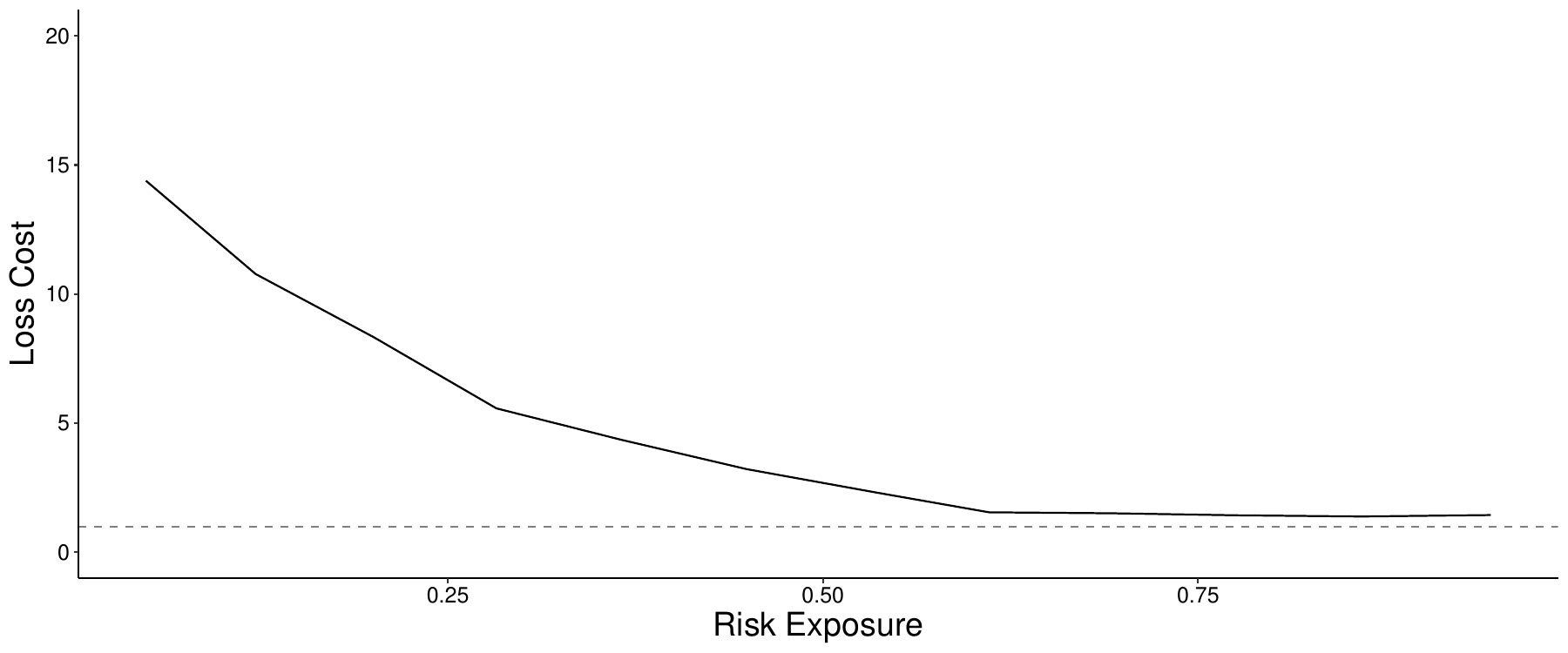}
    \caption{Loss cost references by risk exposures}
    \label{statdesbygroup}
	\end{center}
\end{figure}

For insured individuals with mid-term cancellations, the loss cost reference curves exhibit a clear downward trend as risk exposure increases, which clearly aligns with the Decreasing Scenario in Section \ref{simulation}. This seems to imply that policyholders who cancel at the beginning of their contract have more claims, or claims with greater severity, than others.
This finding highlights the significant role of risk exposure in explaining contract costs. Although traditional pricing models face challenges in predicting an insured’s risk exposure or the likelihood of cancellation before the end of their contract, incorporating this information into pricing models is crucial for accurately estimating parameters and setting premiums.

\subsubsection{ Covariates Description}\label{Covariatesanalysis}

The dataset contains various characteristics for each vehicle and contract. To examine the impact of segmentation on pricing, we selected nine key characteristics as covariates, denoted as \(X_0\), \(X_1\), \(X_2\), \(X_3\), \(X_4\), \(X_5\), \(X_6\), \(X_7\), and \(X_8\) for confidentiality purposes. \(X_0\) represents a covariate with a constant value of 1 across all contracts, while the other covariates (\(X_1\) through \(X_8\)) are binary variables. These covariates capture conventional risk factors such as the insured’s gender, age, vehicle usage, and type of vehicle, in line with standard actuarial practices. To simplify the model, we excluded risk factors that are not typically used in pricing models. As such, we believe the resulting pricing model reflects standard industry practices. Descriptive statistics for these covariates are provided in Figure~\ref{cov.stat}.

\begin{figure}[H]
	\begin{center}
		\includegraphics[scale=0.6]{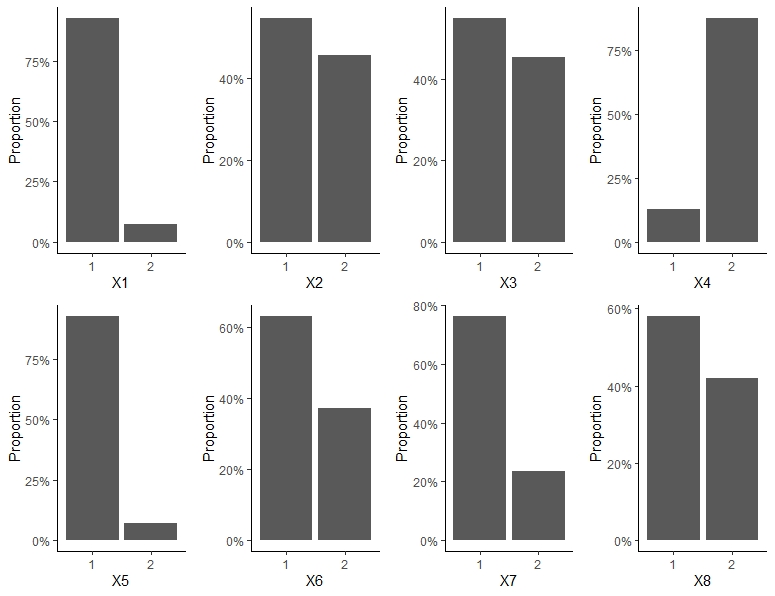}
	\caption{Descriptive statistics of all 8 covariates from the database}
		\label{cov.stat}
	\end{center}
\end{figure}

\subsection{Offset and Ratio Approaches}\label{ap.num.}

To compare the offset and ratio regression approaches, we estimate the parameters using a Tweedie distribution for each segment of the portfolio (insureds with full exposure and those with mid-term cancellation). Furthermore, for each approach (ratio and offset), we calculate the premium for each contract using both methods. It is worth noting that our dataset is a sample from \cite{raissaboucher2024}. Accordingly, we applied the same estimated variance parameter, \( p = 1.42 \), as used in their study for premium modelling. As discussed in Section \ref{rap.twee},  the dispersion parameter does not affect premium modelling; thus, we do not emphasize its role in this analysis.

\subsubsection{Estimated Parameters}\label{ap.num.2.1}

The estimated parameters $\bm{\beta}$ for covariates \(X_1 - X_8\) can be different for both approaches.  To compare the estimated values, we compute a coefficient ratio for each covariate by dividing the estimated coefficient from the offset approach by the corresponding coefficient from the ratio approach.  If the estimated values were identical for both approaches, this ratio would be 1 for all covariates. These results are displayed in Figure \ref{estimparam}. Without surprise, the estimated coefficients are identical in both the offset and ratio approaches, which was expected because the offset and ratio approaches yield equivalent results when the risk exposure is equal to 1.

On the other hand, for contracts with mid-term cancellations, the estimated coefficients differ noticeably between the offset and ratio approaches. This divergence is anticipated because, as discussed in Section \ref{ofwe.twee}, the log-likelihood and gradient functions differ for the two approaches when risk exposure is less than 1. The significant differences, particularly in coefficients associated with covariates \( X_1\), \( X_2\), and \( X_4\), highlight how the choice of approach impacts premium estimation. These differences imply that premiums for certain policyholder profiles will vary substantially based on the chosen approach. 

\begin{figure}[h]
  \centering
    \includegraphics[width=0.75\textwidth]{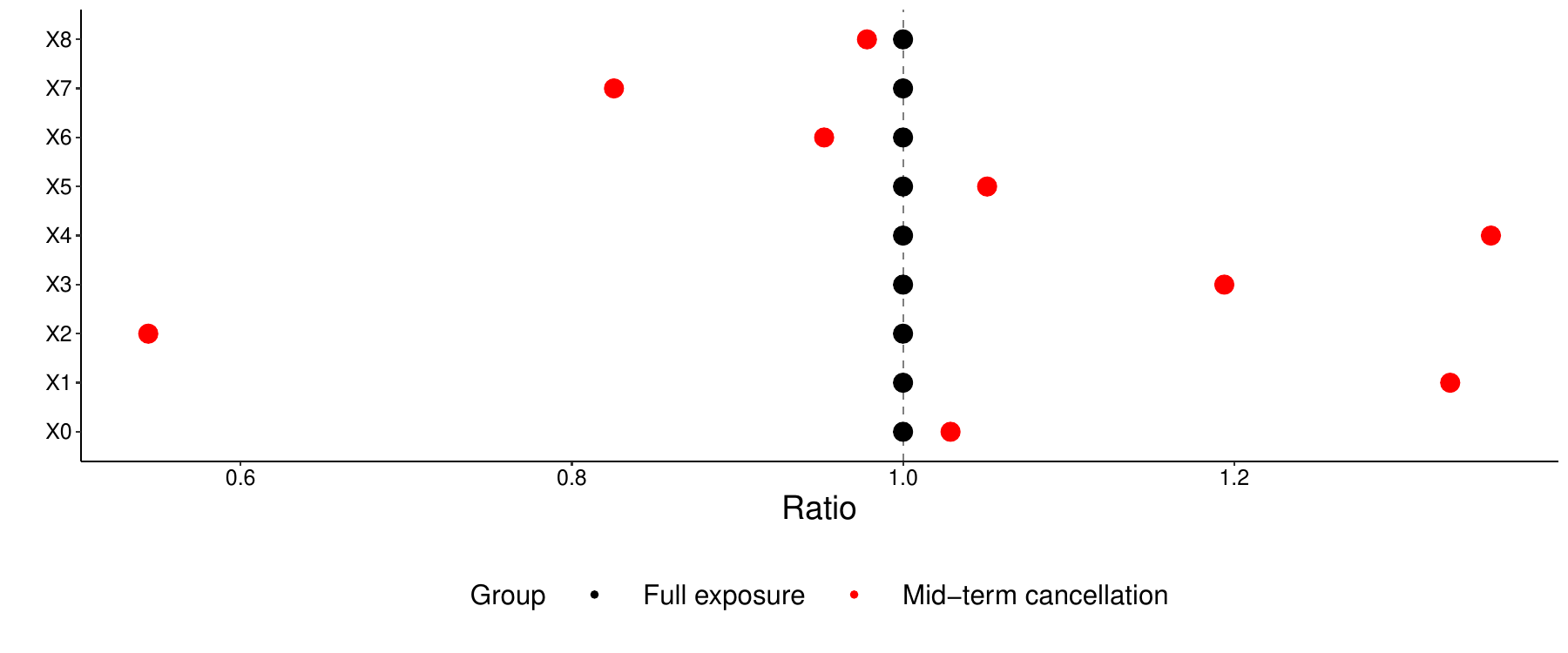}
    \caption{Ratio of estimated parameter vectors}
    \label{estimparam}
\end{figure}

\subsubsection{Estimated Premiums}\label{ap.num.2.1b}

We also compared the distribution of premiums for each approach. To do so, we calculated the premium ratio for each method by dividing the estimated premium for each contract under the offset approach by the premium for the same contract under the ratio approach. The distribution of the ratio of estimated premiums, with box-plots, is presented in Figure \ref{estimpremium}. For the group of insured individuals covered for a one-year period, we observe identical estimated premiums for all contracts in this group (Figure \ref{estimpremium}). This outcome is expected, as all contracts in this group have a risk exposure of 1, and, as previously discussed, the offset and ratio approaches yield equivalent results in this situation.  In contrast, for contracts with mid-term cancellations, we observe premium ratios greater than 1, indicating that the offset approach produces higher premium estimates compared to the ratio approach. This finding is consistent with the Decreasing Scenario discussed in Section \ref{simulation}, where the average claim amount decreases as the average risk exposure increases. 

\begin{figure}[h]
  \centering
      \includegraphics[width=0.75\textwidth]{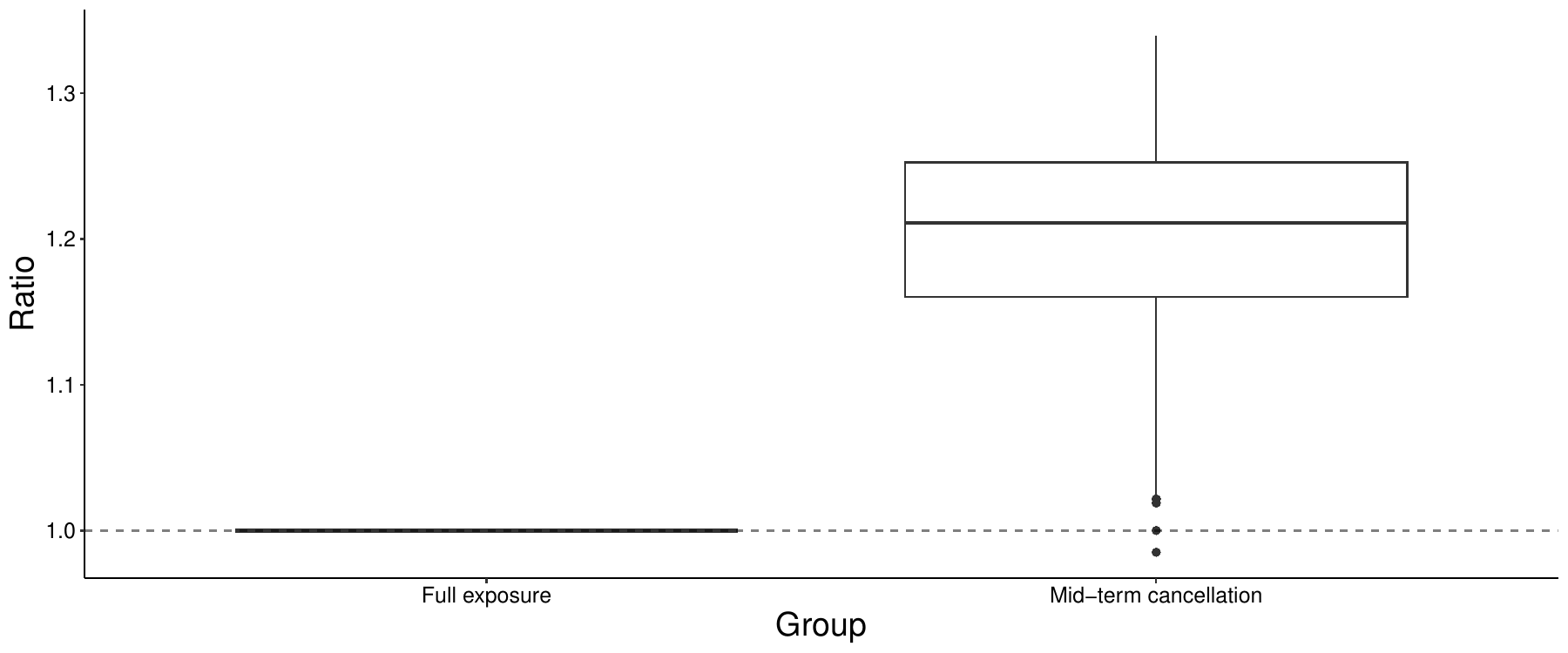}
    \caption{Box-plot of the ratio of estimated premiums}
    \label{estimpremium}
\end{figure}

\subsubsection{Financial Equilibrum Analaysis}\label{ap.num.2.3}

Financial equilibrium requires that the aggregate observed loss costs closely align with the sum of estimated premiums at each level of the risk factors in the regression model. To assess this, we calculated the aggregate observed loss costs and the sum of estimated premiums for each level of our eight risk factors, across all contract groups and approaches. These values were then organized in ascending order based on the aggregate observed loss costs, and we computed ratios by dividing the sum of estimated premiums (from both the offset and ratio approaches) by the aggregate observed loss costs. These ratios are plotted against aggregate observed loss costs in Figure \ref{balance}, with risk classes labeled as integers from 1 to 17.

In the full-exposure group, the sum of individual observed gaps is identical for both the offset and ratio approaches, which is explained by the fact that the estimated premiums are identical under the two methods. As noted previously, it is nevertheless interesting to observe that, in the heterogeneous data, the sums of predicted and observed losses by risk class may differ even when the risk exposure equals one. A similar phenomenon is observed under homogeneous data when the offset approach is used. For the mid-term cancellation group, all ratios are reported in the right panel of Figure~\ref{balance}. The ratio approach yields a closer alignment between the sum of estimated premiums and the aggregate observed loss costs than the offset approach. In particular, the ratio of the sum of estimated premiums to the aggregate observed loss cost under homogeneous data using the offset approach is equal to $1.22$, whereas the ratio approach yields exact aggregate balance under homogeneity. As discussed in Section~\ref{simulation}, ratios of similar magnitude are observed for both approaches in the heterogeneous case.

\begin{figure}[H]
	\begin{center}
		\includegraphics[scale=0.45]{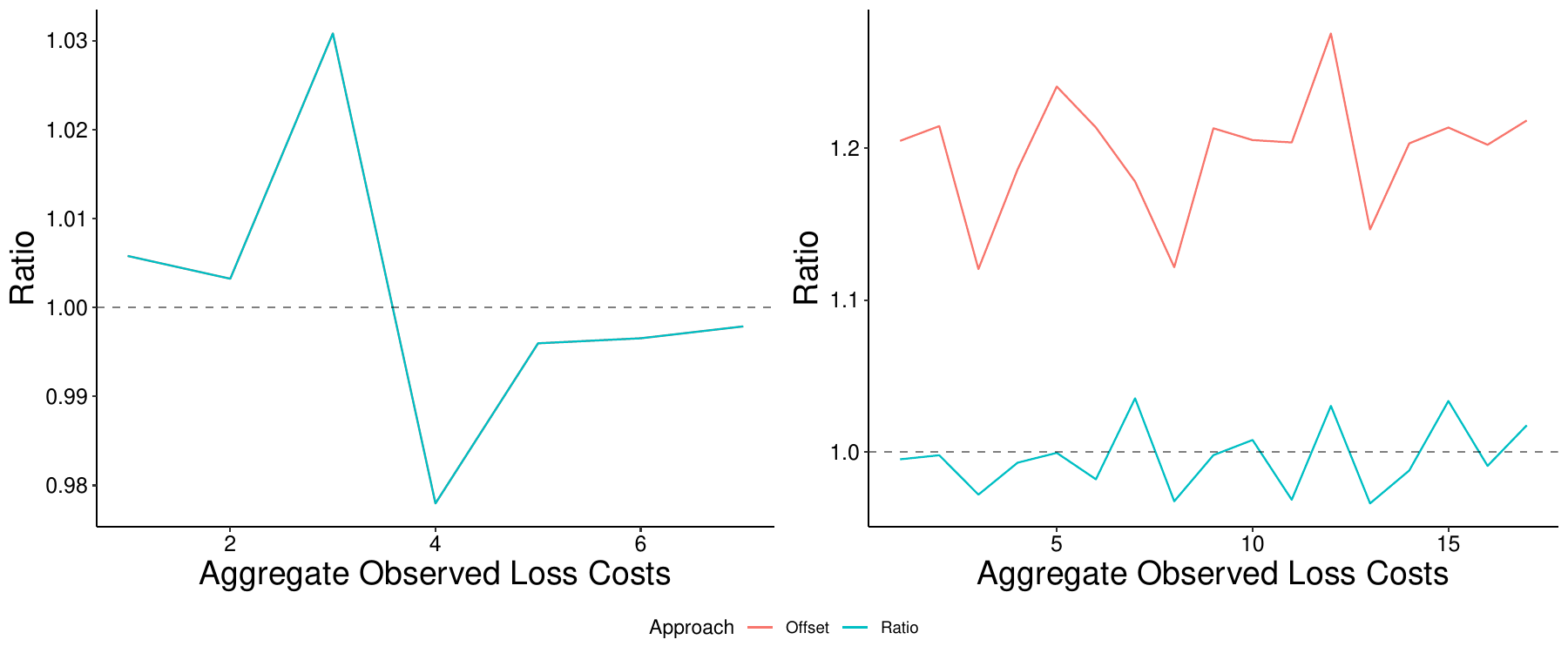}
		\caption{Sum of individual observed gaps comparison (left: Full exposure, right: Mid-term cancellation)}
		\label{balance}
	\end{center}
\end{figure}

\subsection{Discussion} \label{discussion}

Even though neither of the two approaches studied is perfect in every respect, the advantages of using the ratio approach for a Tweedie regression are clear for the insurance portfolio we analyzed. However, the analysis of the modeling choices considered in this paper also brings to light several important issues, which may lead to different conclusions or modeling strategies depending on the portfolio structure and the underlying risk characteristics.

First, the relevance of the offset or ratio approach strongly depends on the composition of the insurance portfolio. As illustrated in Figure~\ref{weightsgraph}, the impact of the chosen approach remains relatively limited for portfolios characterized by a low frequency of mid-term cancellations, such as home insurance or certain commercial insurance lines, where most contracts exhibit full-year exposures. In such contexts, the offset and ratio approaches yield similar results. In contrast, for portfolios with a substantial proportion of partial-year exposures—typical of certain insurance markets, such as North American automobile insurance—the choice between the two approaches becomes critical and can materially affect the modeling outcomes.

These differences are not only technical but also have direct implications for premium determination and policyholder fairness. As shown by the estimated parameters reported in Figure~\ref{estimparam}, the offset and ratio approaches may lead to substantially different premiums for the same insured risk. In a competitive insurance market, such discrepancies may influence which policyholder profiles are favored or disadvantaged by each approach. Understanding how these pricing differences affect fairness perceptions and market competitiveness is therefore an important consideration when selecting a ratemaking methodology.

The empirical results further emphasize the central role of mid-term cancellations in shaping loss costs. Table~\ref{ratio.losscost} and Figure~\ref{truncdesc} show that cancellations have a significant impact on observed loss costs, underlining the need to address this phenomenon more explicitly in pricing models. Rather than incorporating exposure mechanically within a single Tweedie framework, as done in this paper, a more detailed analysis of cancellation behavior could allow insurers to develop more refined modeling strategies. In particular, distinguishing between different causes of cancellations may help capture heterogeneous risk dynamics and improve premium adequacy.

Finally, the analysis raises broader questions regarding the assumption of strict linear proportionality between exposure and expected loss cost, an assumption that underlies both the offset and ratio approaches. While this assumption is widely adopted in actuarial practice, the descriptive results presented in Figure~\ref{statdesbygroup} suggest that loss costs do not necessarily scale linearly with exposure. Shorter exposure periods, in particular, appear to be associated with systematically different loss dynamics, consistent with several scenarios discussed in Section~\ref{Contractdesc}. Although both approaches considered here impose linearity by construction, these empirical patterns indicate that relaxing this assumption could provide additional modeling flexibility and lead to a more accurate representation of exposure-related heterogeneity.

\section{Conclusion}\label{concl}

This paper has examined two widely used strategies for incorporating risk exposure into premium estimation under Tweedie regression models: the offset approach and the ratio approach. Although both methods rely on the same proportionality assumption between exposure and expected loss, we have shown that they induce fundamentally different weighting structures and therefore lead to distinct statistical and financial properties. From a purely inferential perspective, both approaches yield consistent estimators of the true parameter vector under mild regularity conditions. A first-order asymptotic analysis based on the leading Gaussian term of the Edgeworth expansion suggests that the offset approach may exhibit greater efficiency, in the sense of a smaller asymptotic covariance matrix. Simulation results seem to confirm this comparison, indicating that in finite samples the offset approach also displays smaller empirical variance. Beyond parameter estimation, this paper emphasizes that premium modeling in insurance is inherently a financial exercise, for which aggregate balance considerations play a central role. From this standpoint, a key contribution of this work is to show that the ratio approach exhibits a structurally stronger proximity to financial equilibrium. In homogeneous portfolios, it satisfies an exact aggregate balance condition by construction, while in heterogeneous settings it remains systematically closer to balance than the offset approach. Importantly, this property is not driven by asymptotics or incidental compensation effects, but follows directly from the structure of the estimating equations. In the final section of the paper, we complement the theoretical and simulation-based analyses with an empirical study based on a real Canadian automobile insurance portfolio. This application allows for a direct comparison of the offset and ratio approaches both in terms of parameter estimates—where the two approaches yield different values—and in terms of aggregate financial balance. The empirical results indicate that the ratio approach achieves a closer alignment between the sum of estimated premiums and the aggregate observed loss costs, whereas the offset approach exhibits a systematic imbalance.

The empirical results also suggest that a more refined examination of the relationship between risk exposure and observed loss costs may be warranted. In particular, the assumption of a strictly linear proportionality between exposure and expected loss appears to be restrictive in the presence of mid-term cancellations and heterogeneous exposure patterns, even though such proportionality is implicitly imposed by both modeling approaches considered in this paper. One important source of this heterogeneity arises when vehicles are replaced during the policy term, either following the sale and purchase of another vehicle or, as discussed in Section~\ref{ap.num}, after a significant loss that renders the original coverage unnecessary. In this context, more detailed tracking of vehicle replacements within insurers’ databases could help improve the representation of exposure-related risk. More generally, these findings indicate that, while risk exposure $t$ remains a vehicle-level quantity, its relationship with risk may depend on the broader exposure configuration of the policy to which the vehicle belongs. In particular, when multiple vehicles are insured sequentially or concurrently under the same policy, considering interactions among vehicle-level exposures may improve the representation of the underlying risk structure.

\bibliography{bibtex}

@book{boyd2004convex,
  title={Convex optimization},
  author={Boyd, Stephen P and Vandenberghe, Lieven},
  year={2004},
  publisher={Cambridge University Press}
}

@article{delong2021making,
  title={Making {Tweedie}’s compound {Poisson} model more accessible},
  author={Delong, {\L}ukasz and Lindholm, Mathias and W{\"u}thrich, Mario V},
  journal={European Actuarial Journal},
  volume={11},
  number={1},
  pages={185--226},
  year={2021},
  publisher={Springer}
}

@book{denuit2019effective,
  title={Effective statistical learning methods for actuaries I},
  author={Denuit, Michel and Trufin, Julien},
  year={2019},
  publisher={Springer, Cham }
}

@book{frees2014predictive,
  title={Predictive modeling applications in actuarial science},
  author={Frees, Edward W and Derrig, Richard A and Meyers, Glenn},
  volume={1},
  year={2014},
  publisher={Cambridge University Press}
}

@book{horn2012matrix,
  title={Matrix analysis},
  author={Horn, Roger A and Johnson, Charles R},
  year={2012},
  publisher={Cambridge University Press}
}

@book{denuit2007actuarial,
  title={Actuarial modelling of claim counts: Risk classification, credibility and bonus-malus systems},
  author={Denuit, Michel and Mar{\'e}chal, Xavier and Pitrebois, Sandra and Walhin, Jean-Fran{\c{c}}ois},
  year={2007},
  publisher={Wiley, West Sussex}
}

@book{wuthrich2023statistical,
  title={Statistical foundations of actuarial learning and its applications},
  author={W{\"u}thrich, Mario V and Merz, Michael},
  year={2023},
  publisher={Springer, Cham}
}

@article{raissaboucher2024, 
title={Bonus-Malus Scale premiums for {Tweedie}’s compound {Poisson} models}, 
DOI={10.1017/S1748499524000113}, 
journal={Annals of Actuarial Science}, 
author={Boucher, Jean-Philippe and Coulibaly, Raïssa}, 
pages={1--25},
year={2024}
}

@article{nelder1972generalized,
  title={Generalized linear models},
  author={Nelder, John Ashworth and Wedderburn, Robert WM},
  journal={Journal of the Royal Statistical Society Series A: Statistics in Society},
  volume={135},
  number={3},
  pages={370--384},
  year={1972},
  publisher={Oxford University Press}
}

@article{jorgensen1987exponential,
  title={Exponential dispersion models},
  author={J{\o}rgensen, Bent},
  journal={Journal of the Royal Statistical Society Series B: Statistical Methodology},
  volume={49},
  number={2},
  pages={127--145},
  year={1987},
  publisher={Oxford University Press}
}

@inproceedings{Tweedie1984index,
  title={An index which distinguishes between some important exponential families},
  author={{Tweedie}, Maurice CK and others},
  booktitle={Statistics: Applications and new directions: Proc. Indian statistical institute golden Jubilee International conference},
  volume={579},
  pages={579--604},
  year={1984}
}

@book{jorgensen1997theory,
  title={The theory of dispersion models},
  author={Jorgensen, Bent},
  year={1997},
  publisher={CRC Press}
}

@book{casella2024statistical,
  title={Statistical inference},
  author={Casella, George and Berger, Roger},
  year={2024},
  publisher={CRC Press}
}

@article{denuit2024testing,
  title={Testing for auto-calibration with {Lorenz} and Concentration curves},
  author={Denuit, Michel and Huyghe, Julie and Trufin, Julien and Verdebout, Thomas},
  journal={Insurance: Mathematics and Economics},
  volume={117},
  pages={130--139},
  year={2024},
  publisher={Elsevier}
}

@article{white1982maximum,
  title={Maximum likelihood estimation of misspecified models},
  author={White, Halbert},
  journal={Econometrica: Journal of the econometric society},
  pages={1--25},
  year={1982},
  publisher={JSTOR}
}

@book{kullback,
  title={Entropy and information theory},
  author={Gray, Robert M},
  year={2011},
  publisher={Springer Science \& Business Media}
}
	
\newpage

\setcounter{section}{0}
\setcounter{subsection}{0}

\section*{Appendix 1: Calculation of Asymptotic Covariance Matrices}\label{appendix2}

This appendix summarizes results from \citet{white1982maximum} that are directly relevant to Proposition \ref{proposition.00}. We provide additional details on the calculation of the asymptotic covariance matrices of the parameter estimators in the Offset and Ratio approaches. To this end, we begin by recalling the expression of the response variance in both approaches, which is fundamental for deriving the covariance matrices:

$$ \Var{Z_i \mid \mathbf{X}_i,w_i} = \frac{\phi}{w_i} \zeta_i^p, \quad i=1,\ldots,n,$$

where \( \zeta_i = \exp\{\mathbf{X}_i^\top \bm{\beta}\}\), and $w_i = t_i^{2-p}$ under the Offset approach while $w_i = t_i$ under the ratio approach. 

For convenience, we also introduce the notation
\[
\bm{D}(W) = \diag\!\left(w_i \zeta_i^{2-p}\right)_{i=1,\ldots,n},
\]
with the vector $\bm{W} = (w_1,\ldots,w_n)^\top$. 
This expression for the variance is used in the computation of the asymptotic covariance matrix of the parameter estimator. 
Following \cite{white1982maximum}, this matrix is defined as
\[
\Sigma = A^{-1}BA^{-1},
\]
where $A = - \tfrac{1}{\phi}\bm{X}^\top \bm{D}(W)\bm{X}$ and 
$B = \mathbb{E}\!\left[\nabla(\bm{\beta}) \nabla(\bm{\beta})^{\top}\right]$, 
with the score vector
\[
\nabla(\bm{\beta}) = \frac{1}{\phi} \sum_{i=1}^n w_i \frac{Z_i - \zeta_i}{\zeta_i^{p-1}} \mathbf{X}_i.
\]

Then, the matrix $B$ is
\begin{align*}
   B &= \frac{1}{\phi^2} \Esp{ 
   \left( \sum_{i=1}^n w_i \frac{Z_i - \zeta_i}{\zeta_i^{p-1}} x_{i,j_1} \right)
   \left( \sum_{i=1}^n w_i \frac{Z_i - \zeta_i}{\zeta_i^{p-1}} x_{i,j_2} \right)}_{j_1,j_2=0,\ldots,q}  \\
   &= \frac{1}{\phi^2} \Esp{
   \sum_{i=1}^n w_i^2 \frac{(Z_i - \zeta_i)^2}{\zeta_i^{2p-2}} 
   x_{i,j_1} x_{i,j_2}}_{j_1,j_2=0,\ldots,q}  \\
   &= \frac{1}{\phi^2} \sum_{i=1}^n w_i^2 \frac{\Var{Z_i \mid \mathbf{X}_i,w_i}}{\zeta_i^{2p-2}} 
   x_{i,j_1} x_{i,j_2}  \\
   &= \frac{1}{\phi^2} \sum_{i=1}^n w_i^2 \frac{\tfrac{\phi}{w_i}\zeta_i^p}{\zeta_i^{2p-2}} 
   x_{i,j_1} x_{i,j_2}  \\
   &= \frac{1}{\phi} \sum_{i=1}^n w_i \zeta_i^{2-p} x_{i,j_1} x_{i,j_2}  \\
   &= \frac{1}{\phi}\bm{X}^\top \bm{D}(W) \bm{X}  \\
   &= -A. 
\end{align*}

Hence, $B = -A$, and the asymptotic covariance matrix reduces to
\[
\Sigma = -A^{-1}.
\]

\end{document}